\documentclass[twocolumn,10pt,twosided]{IEEEtran}
\usepackage{amssymb}
\usepackage{epstopdf}
\usepackage{multirow}
\usepackage{times,amsmath,amssymb,graphicx,epsf,psfrag,epsfig,cite,color,mathrsfs,url,float}
\usepackage{graphicx}
\usepackage{algpseudocode}
\usepackage{algorithm}

\newtheorem{corollary}{Corollary}
\newtheorem{theorem}{Theorem}
\newtheorem{proposition}{Proposition}
\newtheorem{lemma}{Lemma}
\newtheorem{fact}{Fact}
\newtheorem{remark}{Remark}
\newtheorem{definition}{Definition}
\newtheorem{observation}{Observation}
\newcommand{\eg}{\emph{e.g., }}

\newcommand{\ie}{\emph{i.e., }}
\makeatletter
\newcommand\figcaption{\def\@captype{figure}\caption}
\newcommand\tabcaption{\def\@captype{table}\caption}
\makeatother
\begin{document}
\ifodd 1
\newcommand{\revh}[1]{{\color{magenta}#1}}
\newcommand{\rev}[1]{{\color{black}#1}} 
\newcommand{\com}[1]{\textbf{\color{red} (COMMENT: #1)}} 
\newcommand{\response}[1]{\textbf{\color{green} (RESPONSE: #1)}} 
\else
\newcommand{\revh}[1]{#1}
\newcommand{\rev}[1]{#1}
\newcommand{\com}[1]{}
\newcommand{\response}[1]{}
\fi

\title{Price Differentiation for Communication Networks}

\author{Shuqin Li, \emph{Student Member, IEEE},  and Jianwei Huang, \emph{Senior Member, IEEE} \thanks{\IEEEcompsocthanksitem Shuqin Li is with Bell Labs Shanghai, Alcatel-Lucent Shanghai Bell Co., Ltd., D400, Bldg. 3, 388 Ningqiao Rd, Shanghai, 201206, China.
E-mail: Shuqin.Li@alcatel-sbell.com.cn This work was done when Shuqin Li was in The Chinese University of Hong Kong.
\IEEEcompsocthanksitem Jianwei Huang is with Department of Information Engineering, The  Chinese University of Hong Kong. E-mail: jwhuang@ie.cuhk.edu.hk.
\IEEEcompsocthanksitem Part of the results have appeared in \emph{IEEE GLOBECOM} 2009\cite{li2009revenue}. This work is supported by the General Research Funds (Project Number 412710 and 412511) established under the University Grant Committee of the Hong Kong Special Administrative Region, China.
}}

\maketitle

\begin{abstract}
We study the optimal usage-based pricing problem in a resource-constrained network with one profit-maximizing \rev{service provider} and multiple groups of surplus-maximizing users. With the assumption that the \rev{service provider} knows the utility function of each user (thus complete information), we find that the complete price differentiation scheme can achieve a large revenue gain (e.g., 50\%) compared to no price differentiation, when the total network resource is comparably limited and the high willingness to pay users are minorities. However, the complete price differentiation scheme may lead to a high implementational complexity. To trade off the revenue against the  implementational complexity, we further study the partial price differentiation scheme, and design a polynomial-time algorithm that can compute the optimal partial differentiation prices. We also consider the incomplete information case where the service provider does not know which group each user belongs to. We show that it is still possible to realize price differentiation under this scenario, and provide the sufficient and necessary condition under which an incentive compatible differentiation scheme can achieve the same revenue as under complete information.
\end{abstract}

\section{Introduction}

Pricing is important for the design, operation, and management of communication networks. Pricing has been used with two different meanings in the area of communication networks.  One is the ``optimization-oriented'' pricing for network resource allocation: it is made popular by Kelly's seminal work on network congestion control\cite{kelly1997charging,kelly1998rate}. For example, the  Transmission Control Protocol (TCP) has been successfully reverse-engineered as a congestion pricing based solution to a network optimization problem \cite{low1999optimization,kunniyur2003end}. A more general framework of Network Utility Maximization (NUM) was subsequently developed to forward-engineer many new network protocols (see a recent survey in \cite{chiang2007layering}). In various NUM formulations, the ``optimization-oriented'' prices often represent the Lagrangian multipliers of various resource constraints and are used to coordinate different network entities to achieve the maximum system performance in a distributed fashion. The other is the ``economics-based'' pricing, which is used by a network service provider to various objectives including revenue maximization. The proper design of such a pricing becomes particularly challenging today due to the exponential growth of data volume and applications in both wireline and wireless networks. In this paper, we focus on studying the ``economics-based'' pricing schemes for managing communication networks.

Economists have proposed many sophisticated pricing mechanisms to extract surpluses from the consumers and maximize revenue (or profits) for the providers. A typical example is the optimal nonlinear pricing\cite{mussa1978monopoly, stokey1979intertemporal, maskin1984monopoly}.  In practice, however, we often observe simple pricing schemes deployed by the service providers. Typical examples include flat-fee pricing and (piecewise) linear usage-based pricing. One potential reason behind the gap between ``theory'' and ``practice'' is that the optimal pricing schemes derived in economics often has a high implementational complexity. Besides a higher maintenance cost,  complex pricing schemes are not ``customer-friendly'' and discourage customers from using the services \cite{shakkottai2008price,valancius2011many}.  Furthermore, achieving the highest possible revenue often with complicated pricing schemes requires knowing the information (identity and preference) of each customer, which can be challenging in large scale communication networks. It is then natural to ask the following two questions:
\begin{enumerate}
    \item  How to design simple pricing schemes to achieve the best tradeoff between complexity and performance?
	 \item  How does the network information structure impact the design of pricing schemes?
\end{enumerate}

This paper tries to answer the above two questions with some stylized communication network models.
Different from some previous work that considered a flat-fee pricing scheme where the payment does not depend on the resource consumption (\eg \cite{acemoglu2004marginal,gibbens2000internet, shakkottai2008price}), here we study the revenue maximization problem with the linear usage-based pricing schemes, where a user's total payment is linearly proportional to allocated resource. In  wireless communication networks, however, the usage-based pricing scheme seems to become increasingly popular due to the rapid growth of wireless data traffic. In June 2010,  AT$\&$T in the US switched from the flat-free based pricing (\ie unlimited data  for a fixed fee) to the  usage-based pricing schemes for 3G wireless data \cite{news1}.  Verizon followed up with similar plans in July 2011. Similar usage-based pricing plans have been adopted by major Chinese wireless service providers including China Mobile and China UniCom.
Thus the research on the usage-based pricing is of great practical importance.

In this paper,  we consider the revenue maximization problem of a monopolist service provider facing multiple groups of users. Each user determines its optimal resource demand to maximize the surplus, which is the difference between its utility and payment. The service provider chooses the pricing schemes to maximize his revenue, subject to a limited resource. We consider both complete information and incomplete information scenarios and design different pricing schemes with different implementational complexity levels.

Our main contributions are as follows.
\begin{itemize}
    \item \emph{Complete network information}: We propose a polynomial algorithm to compute the optimal solution of the partial price differentiation problem, which includes the complete price differentiation scheme and the single pricing scheme as special cases. The optimal solution has a threshold structure, which allocates positive resources to high willingness to pay users with priorities.
	 \item \emph{Revenue gain under the complete network information}: Compared to the single pricing scheme, we identify the two important factors behind the revenue increase of the (complete and partial) price differentiation schemes: the differentiation gain and the effective market size. The revenue gain is the most significant when high users are minority among the whole population and total resource is limited but not too small.
	 \item \emph{Incomplete network information}: We design an incentive-compatible scheme with the goal to achieve the same maximum revenue that can be achieved with the complete information. We find that if the differences of willingness to pays of users are larger than some thresholds, this incentive-compatible scheme can achieve the same maximum revenue. We further characterize the necessary and sufficient condition for the thresholds.
\end{itemize}

It is interesting to compare our results under the complete network information scenario with results in \cite{shakkottai2008price} and \cite{chau2010viability}. In \cite{shakkottai2008price}, the authors showed that the revenue gain of price differentiation is small with a flat entry-fee based  Paris Metro Pricing (\eg \cite{odlyzko1999paris}), and a complicated differentiation strategy may not be worthwhile. Chau et al.\cite{chau2010viability} further derived the sufficient conditions of congestion functions that guarantee the viability of these Paris Metro Pricing schemes.
By contrast, our results show that the revenue gain of price differentiation can be substantial for a usage-based pricing system.

Some recent work of usage-based pricing and revenue management in communication network includes \cite{basar2002revenue, shen2007optimal, daoud2008stackelberg, jiang2008time, hande2010pricing, he2005pricing,shakkottai2006economics, voja}. Basar and Srikant in \cite{basar2002revenue} investigated the bandwidth allocation problem in a single link network with the single pricing scheme. Shen and Basar in \cite{shen2007optimal} extended the study to a more general nonlinear pricing case with the incomplete network information scenario. They discussed the single pricing scheme under incomplete information with a continuum distribution of users' types. In contrast, our study on the incomplete information focuses on the linear pricing with a discrete setting of users' types. We also show that, besides the single pricing scheme, it is also possible to design differentiation pricing schemes under incomplete information. Daoud \emph{et al.}  \cite{daoud2008stackelberg} studied a uplink power allocation problem in a CDMA system, where the interference among users are the key constraint instead of the limited total resource considered in our paper.
Jiang et al. in \cite{jiang2008time} and Hande et al. in \cite{hande2010pricing}  focused on the study of the time-dependent pricing. He and Walrand in  \cite{he2005pricing},  Shakkottai and Srikant in \cite{shakkottai2006economics} and Gajic et al. in \cite{voja} focused on the interaction between different service providers embodied in the pricing strategies, rather than the design of the pricing mechanism. Besides, none of the related work considered the partial differential pricing as in our paper.

%

\section{System Model}
We consider a network with a total amount of $S$ limited resource (which can be in the form of rate, bandwidth, power, time slot, etc.). The resource is allocated by a monopolistic \rev{service provider} to a set $\mathcal{I}=\{1,\ldots,I\}$ of user groups.  Each group $i\in\mathcal{I}$ has $N_i$ homogeneous users\footnote{A special case is $N_{i}$=1 for each group, i.e., all users in the network are different.} with the same utility function:
\begin{equation}
    u_i(s_i)=\theta_i \ln(1+s_i),
\label{eq:log_u}
\end{equation}
where $s_i$ is the allocated resource to one user in group~$i$ and $\theta_i$ represents the willingness to pay of group $i$. \rev{The logarithmic utility function is commonly used to model the proportionally fair resource allocation in communication networks (see \cite{basar2002revenue} for detailed explanations).
The analysis of the complete information case can also be extended to more general utility functions (see Appendix~\ref{sec:general_utility}). } Without loss of generality, we assume that $\theta_1>\theta_2>\dots>\theta_I$. In other words, group~$1$ contains users with the highest valuation, and group~$I$ contains users with the lowest valuation.

We consider two types of information structures:
\begin{enumerate}
    \item \textbf{Complete information}: the \rev{service provider} knows each user's utility function.
 Though the complete information is a very strong assumption, it is the most frequently studied scenario the network pricing literature \cite{basar2002revenue,shen2007optimal,daoud2008stackelberg,jiang2008time, hande2010pricing, he2005pricing, shakkottai2006economics, voja}. The significance of studying the complete information is two-fold. It serves as the benchmark of practical designs and provides important insights for the incomplete information analysis.

	 \item \textbf{Incomplete information}: the \rev{service provider} knows the total number of groups $I$, the number of users in each group $N_{i}, i\in\mathcal{I}$, and the utility function of each group $u_{i}, i\in\mathcal{I}$. It does not know which user belongs to which group.
Such assumption in our discrete setting is analogous to that the service provider knows only the users' types distribution in a continuum case. Such statistical information can be obtained through long term observations of a stationary user population.
\end{enumerate}

The interaction between the \rev{service provider} and users can be characterized as a two-stage Stackelberg model shown in Fig.~\ref{fig:model}.
\begin{figure}[ht]
\centering
\includegraphics[scale=0.45]{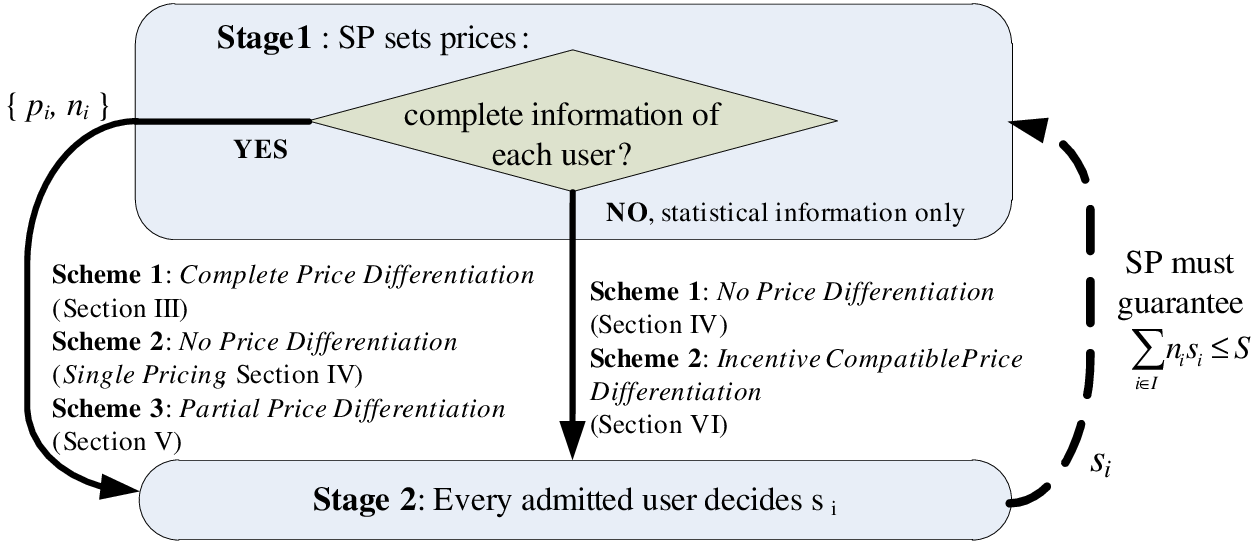}
\caption{A two-stage Stackelberg model} \label{fig:model}
\end{figure}
The service provider publishes the pricing scheme in Stage~1, and users respond with their demands in Stage~2. The users want to maximize their surpluses by optimizing their demands according to the pricing scheme. The service provider wants to maximize its revenue by setting the right pricing scheme to induce desirable demands from users. Since the service provider has a limited total resource, it must guarantee that the total demand from users is no larger than what it can supply.

The details of pricing schemes depend on the information structure of the service provider. Under complete information, since the service provider can distinguish different groups of users, it announces the pricing and the admission control decisions to different groups of users. It can choose from the complete price differentiation scheme, the single pricing scheme, and the partial price differentiation scheme to realize a desired trade-off between the implementational complexity and the total revenue. Under incomplete information, it publishes a common price \emph{menu} to all users, and allow users to freely choose a particular price option in this menu. All these pricing schemes will be discussed one by one in the following sections.

Note that it is possible for a user to achieve an ``arbitrage'' by splitting into several smaller users, each of which requests a small amount of resource and enjoys a lower unit price. Fortunately, preventing arbitrage of services is often easier and less costly than that of goods \cite{pashigian1995price}. For example, we can often uniquely  identify a user through its MAC address. Full discussion of arbitrage prevention, however, is beyond the scope of this paper.

%
\section{Complete Price Differentiation under complete information}
\label{sec:cpd}
We first consider the complete information case. Since the \rev{service provider} knows the utility and the identity of each user, it is possible to maximize the revenue by charging a different price to each group of users. The analysis will be based on backward induction, starting from Stage~2 and then moving to Stage~1.
\subsection{User's Surplus Maximization Problem in Stage 2}
If a user in group $i$ has been admitted into the network and offered a linear price $p_i$ in Stage~1, then it solves the following surplus maximization problem,
\begin{equation}
\label{eq:user's_problem}
    \underset{s_{i}\ge 0}{\text{maximize}}\; u_{i}(s_{i})-p_{i} s_{i},
\end{equation}
which leads to the following unique optimal demand
\begin{equation}
    s_i(p_{i}) =\left(\frac{\theta_i }{p_i} - 1\right)^+,\,\;\text{where}\;(\cdot)^+ \triangleq \max(\cdot,0).
\label{eq:user's_demand}
\end{equation}
\begin{remark}
The analysis of the Stage~2 user surplus maximization problem is the same for all pricing schemes. The result in (\ref{eq:user's_demand}) will be also used in Sections \ref{sec:single_price}, \ref{sec:PPD} and \ref{sec:incomplete_information}.
\end{remark}

\subsection{Service Provider's Pricing and Admission Control Problem in Stage 1}
\label{sec:sp_cpd}
In Stage~1, the service provider maximizes its revenue by choosing  the price $p_i$ and the  number of admitted users $n_i$ for each group $i$ subject to the limited total resource $S$. The key idea is to perform a Complete Price differentiation ($CP$) scheme, i.e., charging each group with a different price.
\begin{eqnarray}
    CP:& \underset{\boldsymbol{p}\geq 0, \boldsymbol{s}\geq 0,\boldsymbol{n}}{\text{maximize}} & \sum\limits_{i\in\mathcal{I}} n_i p_{i} s_{i} \label{eq:cp_obj}\\
    &\text{subject to}& s_i =\left(\frac{\theta_i }{p_i} - 1\right)^+,\;\; i\in\mathcal{I}, \label{demand_function}\\
    && n_i \in\{0,\ldots,N_{i}\} \;\;,\;\;  i\in\mathcal{I}, \label{adimission_control}\\
    &&\sum\limits_{i\in\mathcal{I}} n_i s_{i} \le S. \label{resource_constraint}
\end{eqnarray}
where $\boldsymbol{p}\overset{\Delta}{=}\{p_i, i\in\mathcal{I}\}$, $\boldsymbol{s}\overset{\Delta}{=}\{s_i, i\in\mathcal{I}\}$, and $\boldsymbol{n}\overset{\Delta}{=}\{n_i, i\in\mathcal{I}\}$. We use bold symbols to
denote vectors in the sequel. Constraint~(\ref{demand_function}) is the solution of the Stage~2 user surplus maximization problem in (\ref{eq:user's_demand}). Constraint
(\ref{adimission_control}) denotes the admission control decision, and
constraint (\ref{resource_constraint}) represents the total limited
resource in the network.

The $CP$ problem  is not straightforward to solve, since it is a non-convex
optimization problem with a non-convex objective function (summation of
products of $n_{i}$ and $p_{i}$), a coupled constraint
(\ref{resource_constraint}), and integer variables $\boldsymbol{n}$.
However, it is possible to convert it into an equivalent convex
formulation  through a series of transformations, and thus the problem can be solved efficiently.

First, we can remove the $(\cdot)^+$ sign in constraint (\ref{demand_function}) by realizing the fact that there is no need to set $p_i$ higher than $\theta_i$  for users in group
$i$; users in group $i$ already demand zero resource and
generate zero revenue  when $p_i = \theta_i$. This means that we can rewrite constraint
(\ref{demand_function}) as
\begin{equation}
    p_i=\frac{\theta_i}{s_i+1} \;{\rm{and}}\;s_i\ge 0,  i\in\mathcal{I}.
\label{transform_con}
\end{equation}
Plugging (\ref{transform_con}) into (\ref{eq:cp_obj}), then the objective function becomes $\sum\limits_{i\in\mathcal{I}} n_i \frac{\theta_is_i}{s_i+1}$.
We can further decompose the $CP$ problem in the following two subproblems:
\begin{enumerate}
\item \emph{Resource allocation}: for a fixed admission control decision $\boldsymbol{n}$,  solve for the optimal resource allocation $\boldsymbol{s}$.
\begin{eqnarray}
    CP_1:&\underset{\boldsymbol{s}\geq 0}{ \text{maximize}}\;  & \sum\limits_{i\in\mathcal{I}} n_i \frac{\theta_is_i}{s_i+1}\nonumber\\
    &\text{subject to}\;&\sum\limits_{i\in\mathcal{I}} n_i s_{i} \le S. \label{resource_c}
\end{eqnarray}

Denote the solution of $CP_1$ as $\boldsymbol{s}^* =(s_i^*(\boldsymbol{n}),\forall i\in\mathcal{I})$. We further maximize the revenue of the integer admission control variables $\boldsymbol{n}$.
\item \textit{Admission control}:
    \begin{eqnarray}
    CP_2:& \underset{\boldsymbol{n}}{\text{maximize}}\; & \sum\limits_{i\in\mathcal{I}} n_i \frac{\theta_is_i^*(\boldsymbol{n})}{s_i^*(\boldsymbol{n})+1} \label{revenue_function} \\
    &{\text{subject to}}\;& n_i \in\{0,\ldots,N_{i}\} \;\;,\;\;  i\in\mathcal{I} \nonumber
\end{eqnarray}
\end{enumerate}

Let us first solve the $CP_1$ subproblem   in $\boldsymbol{s}$. Note that it is a convex optimization problem. By using Lagrange multiplier technique, we can get the first-order necessary and sufficient condition:
\begin{equation}
\label{eq:cd_resource}
s^*_{i}(\lambda) = \left(\sqrt{\frac{\theta_i}{\lambda}}-1\right)^+,
\end{equation}
where $\lambda$ \rev{is} the Lagrange multiplier corresponding to the resource constraint (\ref{resource_c}).

Meanwhile, we note the resource constraint (\ref{resource_c}) must hold with
equality, since the objective is strictly increasing function in $s_{i}$. Thus, by plugging (\ref{eq:cd_resource}) into (\ref{resource_c}), we have
\begin{equation}
\label{eq:cd_water_filling}
    \sum\limits_{i\in\mathcal{I}} n_i \left(\sqrt{\frac{\theta_i}{\lambda}}-1\right)^+ = S.
\end{equation}
This weighted water-filling problem (where $\frac{1}{\sqrt{\lambda}}$ can be viewed as the water level) in general has no closed-form solution  for $\lambda$. However, we can efficiently determine the optimal solution $\lambda^*$ by exploiting the special structure of our problem. Note that since $\theta_1>\dots>\theta_I$, then $\lambda^*$ must satisfy  the following condition:

\begin{equation}
\label{eq:water-filling_cpd}
    \sum_{i=1}^{K^{cp}} n_i\left(\sqrt{\frac{\theta_i}{\lambda^*}}-1\right)=S,
\end{equation}
for a group index threshold value $K^{cp}$ satisfying
\begin{equation}
\label{eq:cd_threshold}
\frac{\theta_{K^{cp}}}{\lambda^*}> 1\; {\rm{and}}\; \frac{\theta_{K^{cp}+1}}{\lambda^*}\le 1.
\end{equation}
In other words, only groups with index no larger than $K_{cp}$ will be allocated the positive resource.
This property leads to the following simple Algorithm~\ref{alg:1} to compute $\lambda^*$ and group index threshold $K^{cp}$: we start by assuming $K^{cp}=I$ and compute $\lambda$. If (\ref{eq:cd_threshold}) is not satisfied, we decrease $K^{cp}$ by one and recompute $\lambda$ until (\ref{eq:cd_threshold}) is satisfied.
\begin{algorithm}
\caption{Solving the Resource Allocation Problem $CP_1$}     
\label{alg:1}
\begin{algorithmic}[1]
 \Function {$CP$} {$\{n_i,\theta_i\}_{i\in{\mathcal{I}}}$, $S$}
    \State $k\gets I$,  $\lambda(k) \gets \left(\frac{\sum_{i=1}^{k} n_i \sqrt {\theta_i}}{S+\sum_{i=1}^{k}n_i}\right)^2$
    \While {$\theta_k\le \lambda(k)$}
        \State $k \gets k-1$,  $\lambda(k) \gets \left(\frac{\sum_{i=1}^{k} n_i \sqrt {\theta_i}}{S+\sum_{i=1}^{k}n_i}\right)^2$
    \EndWhile
    \State $K^{cp}\gets k$, $\lambda^*\gets\lambda(k)$
\State \textbf{return} $K^{cp}$, $\lambda^*$
\EndFunction
\end{algorithmic}                       
\end{algorithm}

Since $\theta_1>\lambda(1)=(\frac{n_1}{s+n_1})^2\theta_1$, Algorithm~\ref{alg:1} always converges and returns the unique values of $K^{cp}$ and $\lambda^{\ast}$. The complexity is $\mathcal{O}(I)$, i.e., linear in the number of user groups (not the number of users).

With $K^{cp}$ and $\lambda^*$, the solution of the resource allocation problem can be written as
\begin{equation}
\label{cp_resource_solution}
    s^*_i= \left\{{\begin{array}{cc}
  \sqrt{\frac{\theta_i}{\lambda^*}}-1, & i=1,\dots,K^{cp};\\
  0, & {\rm{otherwise}}.  \\
\end{array}} \right.
\end{equation}
For the ease of discussions, we introduce a new notion of the \emph{effective market}, which denotes all the groups allocated non-zero resource. For the resource allocation subproblem $CP_1$, the threshold $K^{cp}$ describes the size of the effective market. All groups with indices no larger than $K^{cp}$ are \emph{effective group}s, and users in these groups as \emph{effective user}s. An example of the effective market is illustrated in Fig.~\ref{fig:effective_market_concept}.
\begin{figure}[hbt]
\centering
  \includegraphics[scale=0.5]{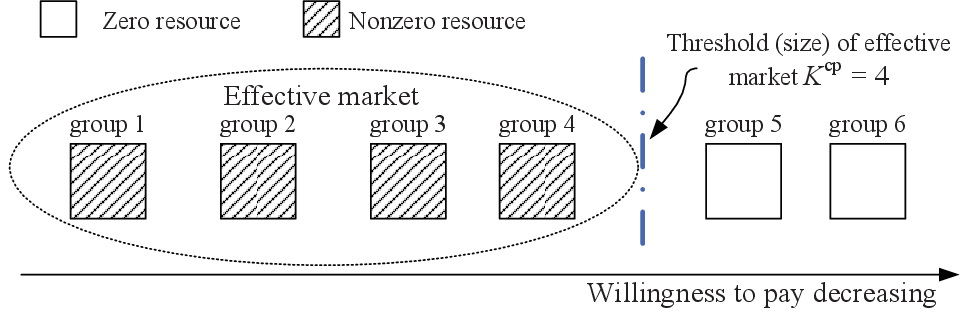}
  \caption{A 6-group example for the effective market:  the willingness to pays decrease from group~1 to group~6. The effective market threshold can be obtained by Algorithm~\ref{alg:1}, and is 4 in this example.}
  \label{fig:effective_market_concept}
\end{figure}

Now let us solve the admission control \rev{subproblem}~$CP_2$. Denote the objective  (\ref{revenue_function}) as $R_{cp}(\boldsymbol{n})$, by (\ref{cp_resource_solution}), then $R_{cp}(\boldsymbol{n})\overset{\Delta}{=}\sum\limits_{i=1}^{K^{cp}}\! n_i\!
\left(\! \sqrt{\frac{\theta_i}{\lambda^*(\boldsymbol{n})}}-1 \!\right)\sqrt{\theta_i\lambda^*\!(\boldsymbol{n})}$. We first relax the integer domain constraint of $n_i$ as $n_i\in [0,N_i]$. Since (\ref{eq:water-filling_cpd}), by taking the derivative of the objective function $R_{cp}(\boldsymbol{n})$ with respect to $n_i$, we have
\begin{align}
\label{eq:cp_n}
\frac{\partial R_{cp}(\boldsymbol{n})}{\partial n_i}= n_i \left(\! \sqrt{\frac{\theta_i}{\lambda^*(\boldsymbol{n})}}-1 \!\right)\frac{\partial \sqrt{\theta_i\lambda^*(\boldsymbol{n})}}{\partial n_i},
\end{align}
Also from (\ref{eq:water-filling_cpd}), we have
$\lambda^*\!\! =\!\! \left(\frac{\sum_{i=i}^{K^{cp}} n_i \sqrt {\theta_i}}{S+\sum_{i=1}^{K^{cp}}n_i}\right)^2$,
thus $\frac{\partial \sqrt{\lambda^*(\boldsymbol{n})}}{\partial n_i}>0$, for $i=1,\dots,K^{cp}$, and $\frac{\partial \sqrt{\lambda^*(\boldsymbol{n})}}{\partial n_i}=0$, for $i=K^{cp}+1,\dots,I$. This means that the objective $R_{cp}(\boldsymbol{n})$
is strictly increasing in $n_i$ for all $i=1,\dots,K^{cp}$, thus it is optimal to admit all users in the effective market. The admission decisions for groups not in the effective market is irrelevant to the optimization, since those users consume zero resource. Therefore, one of the optimal solutions of the
$CP_1$ subproblem is $n^*_i=N_i$ for all $ i\in\mathcal{I}$.
Solving the $CP_1$ and the $CP_2$ subproblems  leads to the optimal solution of  the $CP$ problem:
\begin{theorem}\label{thm:opt}
There exists an optimal solution of th3 $CP$ problem that satisfies the following conditions:
\begin{itemize}
    \item All users are admitted: $n^*_i=N_i$ for all $i\in\mathcal{I}.$
    \item There exist a value $\lambda^\ast$ and a group index threshold $K^{cp} \leq I$, such that only the top $K^{cp}$ groups of users receive positive resource allocations,
$$s^*_i= \left\{{\begin{array}{ll}
  \sqrt{\frac{\theta_i}{\lambda^*}}-1, \;& i=1,\dots,K^{cp};\\
  0, & {\rm{otherwise}}.  \\
\end{array}} \right.$$
with the prices
$$
    p^*_i= \left\{{\begin{array}{ll}
  \sqrt{\theta_i\lambda^*}, \; & i=1,\dots,K^{cp};\\
  \theta_i, & {\rm{otherwise}}.  \\
\end{array}} \right.$$
The values of $\lambda^*$ and $K^{cp}$ can be computed as in Algorithm \ref{alg:1} by setting $n_i=N_i$, for all $i\in\mathcal{I}$.
\end{itemize}
\end{theorem}

Theorem \ref{thm:opt} provides the right economic intuition: the service provider maximizes its revenue by charging a higher price to users with a higher willingness to pay. It is easy to check that $p_i>p_j$ for any $i<j$. The low willingness to pay users are excluded from the markets.

\subsection{Properties}
\label{sub_sec:cpd}
Here we summarize some interesting properties of the optimal complete price differentiation scheme:
\subsubsection{Threshold structure} The threshold based resource allocation means that higher willingness to pay groups have higher priories of obtaining the resource at the optimal solution.

To see this clearly, assume the effective market size is $K^{(1)}$ under parameters $\{\theta_i,N_i^{(1)}\}_{i\in\mathcal{I}}$ and $S$.
Here the superscript (1) denotes the first parameter set.  Now consider another set of parameters $\{\theta_i,N_i^{(2)}\}_{i\in\mathcal{I}}$ and $S$, where $N_i^{(2)}\ge N_i^{(1)}$ for each group $i$ and the new  The effective market size is $K^{(2)}$. By (\ref{eq:water-filling_cpd}), we can see that $K^{(2)}\le K^{(1)}$. Furthermore, we can show that if some high willingness to pay group has many more users under the latter system parameters, i.e., $N_i^{(2)}$ is much larger than $N_i^{(1)}$ for some $i<K^{(1)}$, then the effective size will be strictly decreased, i.e., $K^{(2)}< K^{(1)}$. In other words, the increase of high willingness to pay users will drive the low willingness to pay users out of the effective market.

\subsubsection{Admission control with pricing} Theorem~\ref{thm:opt} shows the explicit admission control is not necessary  at the optimal solution. Also from Theorem~\ref{thm:opt}, we can see that when the number of users in any effective group increases, the price $p^*_i$, for all $i\in\mathcal{I}$ increases and resource $s^*_i$, for all $\forall\;i\le K^{cp}$ decreases. The prices serve as the indications of the scarcity of the resources and will automatically prevent the low willingness to pay users to consume the network resource.

\section{Single Pricing Scheme}
\label{sec:single_price}
In last section, we \rev{showed} that the $CP$ scheme is the optimal pricing scheme to maximize the revenue under complete
information. However, such a complicated pricing scheme is of high implementational complexity. In this section we study the single pricing scheme. It is clear that the scheme will in general suffer a revenue loss comparing with the $CP$ scheme. We will try to characterize the impact of various system parameters on such revenue loss.

\subsection{Problem Formulation and Solution}
Let us first formulate the Single Pricing ($SP$) problem.
\begin{eqnarray*}
    SP:& \underset{p\ge 0,\;\; \boldsymbol{n}}{\text{maximize}} & p\sum\limits_{i\in\mathcal{I}} n_is_{i} \nonumber\\
    &\text{subject to}& s_i =\left(\frac{\theta_i }{p} - 1\right)^+,\;\; i\in\mathcal{I} \label{r_sp}\\
    &&n_i \in\{0,\ldots,N_{i}\} \;\;,\;\;  i\in\mathcal{I}\nonumber\\
    &&\sum\limits_{i\in\mathcal{I}} n_i s_{i} \le S. \label{r_c}
\end{eqnarray*}
Comparing with the $CP$ problem in Section~\ref{sec:cpd}, here the service provider charges a single price $p$ to all groups of users. After a similar transformation as in Section~\ref{sec:cpd}, we can show that the optimal single price satisfies the following the weighted water-filling condition
\begin{equation*}
    \sum\limits_{i\in\mathcal{I}} N_i \left(\frac{\theta_i}{p}-1\right)^+ = S.
\end{equation*}
Thus we can obtain the following solution that shares a similar structure as complete price differentiation.
\begin{theorem}\label{thm:sp}
There exists an optimal solution of the $SP$ problem that satisfies the following conditions:
\begin{itemize}
    \item All users are admitted: $n^*_i=N_i,$ for all $i\in\mathcal{I}.$
    \item There exist a price $p^*$ and a group index threshold $K^{sp} \leq I$, such that only the top $K^{sp}$ groups of users receive positive resource allocations,
$$s^*_i= \left\{{\begin{array}{ll}
  \frac{\theta_i}{p^*}-1, & i=1,2,\dots,K^{sp},\\
  0, & {\rm{otherwise}} , \\
\end{array}} \right.$$
with the price
$$p^*=p(K^{sp})=\frac{\sum_{i=1}^{K^{sp}} N_i \theta_i}{S+\sum_{i=1}^{K^{sp}}N_i}.$$  The value of $K^{sp}$ and $p^*$ can be computed as in Algorithm \ref{alg:2}.
\end{itemize}
\end{theorem}

\begin{algorithm}[htb]                     
\caption{Search the threshold of the $SP$ problem}     
\label{alg:2}
\begin{algorithmic}[1]
 \Function {$SP$} {$\{N_i,\theta_i\}_{i\in{\mathcal{I}}}$, $S$}
    \State $k\gets I$, $p(k)\gets\frac{\sum_{i=1}^{k} N_i {\theta_i}}{S+\sum_{i=1}^{k}N_i}$
    \While {$\theta_k\le p(k)$}
        \State $k \gets k-1$, $p(k)\gets\frac{\sum_{i=1}^{k} N_i {\theta_i}}{S+\sum_{i=1}^{k}N_i}$
    \EndWhile
    \State $K^{sp}\gets k$, $p^*\gets p(k)$
    \State \textbf{return} $K^{sp}$, $p^*$
\EndFunction
\end{algorithmic}                       
\end{algorithm}

\subsection{Properties}
\label{sub_sp}
The $SP$ scheme shares several similar properties as
the $CP$ scheme (Sec.~\ref{sub_sec:cpd}), including the threshold structure and admission control with pricing. Similarly, we can define the effective market for the $SP$ scheme.

It is more interesting to notice the differences between these two schemes.  To distinguish solutions, we use the superscript ``cp'' for the $CP$ scheme, and ``sp'' for the $SP$ scheme.
\begin{proposition}
\label{pro:comparison_sp_cd}
Under same parameters $\{N_i,\theta_i\}_{i\in \mathcal{I}}$ and $S$:
\begin{enumerate}
    \item The effective market of the $SP$ scheme is no larger than the one of the $CP$ scheme, i.e.,
    $K^{sp}\le K^{cp}$.
     \item
There exists a threshold $\bar{k}\in\{1,2\dots,K^{sp}\}$, such that 
\begin{itemize}
    \item  Groups with indices less than $\bar{k}$ ({high} willingness to pay users) are charged with higher prices and allocated less resources in {the $CP$ scheme}, \ie
$p^{cp}_i\ge p^*$ and $s^{cp}_i\le s^{sp}_i$, $\forall\,i\le \bar{k}$,  where the equality holds if only if $i=\bar{k}$ and $\theta_{\bar{k}}=\frac{{p^*}^2}{\lambda^*}$.
     \item Groups with indices greater than $\bar{k}$ ({low} willingness to pay users) are charged with lower prices and allocated more resources {in the $CP$ scheme}, \ie
$p^{cp}_i< p^*$ and $s^{cp}_i> s^{sp}_i$, $\forall\,i\ge \bar{k}$.
\end{itemize}
Here $p^*$ is the optimal single price.
\end{enumerate}
\end{proposition}
\rev{The proof is given in Appendix~\ref{appendix_cp_sp}}. An illustrative example is shown in Fig.~\ref{fig:price} and Fig.~\ref{fig:resource}.
\begin{figure}[ht]
\centering
\includegraphics[scale=0.65]{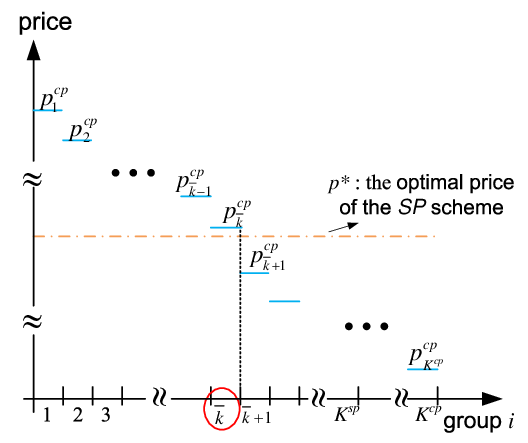}
\caption{The comparison of prices between the $CP$ scheme and the $SP$ scheme} \label{fig:price}
\includegraphics[scale=0.6]{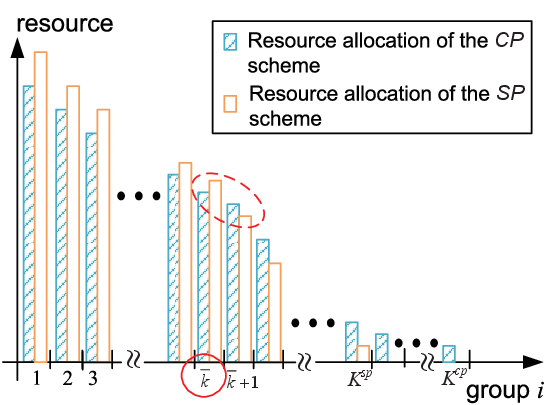}
\caption{The comparison of resource allocation between the $CP$ scheme and the $SP$ scheme}
\label{fig:resource}
\end{figure}

It is easy to understand that the $SP$ scheme makes less revenue, since it is a feasible solution to the $CP$ problem. A little bit more computation sheds more light on this comparison.
We introduce the following notations to streamline the comparison:
\begin{itemize}
    \item $N_{eff}(k) \triangleq \sum\limits_{i=1}^k N_i$: the number of effective users, where $k$ is the size of the effective market.
     \item $\gamma_i(k)\triangleq \frac{N_i}{N_{eff}(k)}$, $i=1,2,\dots,k$: the fraction of group $i$'s users in the  effective market.
     \item $\bar{s}(k)\triangleq \frac{S}{N_{eff}(k)}$: the average resource per an effective user.
     \item $\bar{\theta}(k) \triangleq \sum\limits_{i=1}^k \gamma_i\theta_i$: the average willingness to pay per an effective user.
\end{itemize}
Based on Theorem~\ref{thm:opt},  the revenue of the $CP$ scheme is
\begin{equation}
\label{eq:cp_r_simplified}
R^{cp}(K^{cp})= N_{eff}(K^{cp})\! \left(\frac{\bar{s}(K^{cp})\bar{\theta}(K^{cp})+g(K^{cp})}{\bar{s}(K^{cp})+1}\right)\!,
\end{equation}
where
\begin{equation}
\label{eq:differentiation_gain}
    g(K^{cp})=\frac{1}{\lambda^*}\sum\limits_{i=1}^{K^{cp}}\sum\limits_{j>i}^{K^{cp}} \gamma_i \gamma_j(p^{cp}_i-p^{cp}_j)^2.
\end{equation}
Based on Theorem~\ref{thm:sp}, the revenue of the $SP$ scheme is
\begin{equation}
\label{eq:sp_r_simplified}
R^{sp}(K^{sp})= N_{eff}(K^{sp}) \left(\frac{\bar{s}(K^{sp})\bar{\theta}(K^{sp})}{\bar{s}(K^{sp})+1}\right).
\end{equation}
From (\ref{eq:cp_r_simplified}) and (\ref{eq:sp_r_simplified}), it is clear to see that $R^{cp}\ge R^{sp}$ due to two factors: one is the non-negative term in (\ref{eq:differentiation_gain}), the other is $K^{cp}\ge K^{sp}$: a higher level of differentiation implies a no smaller effective market. Let us further discuss them in the following two cases:
\begin{itemize}
  \item If $K^{cp}=K^{sp}$, then the additional term of (\ref{eq:differentiation_gain}) in (\ref{eq:cp_r_simplified}) means that $R^{cp}\ge R^{sp}$. The equality holds if and only if $K^{cp}=1$, in which case $g(K^{cp})=0$. Note that in this case, the $CP$ scheme degenerates to the $SP$ scheme. We name the nonnegative term $g(K^{cp})$ in
(\ref{eq:differentiation_gain})  as \emph{price differentiation gain}, as it measures the average price difference between any effective users in the $CP$ scheme. The larger the price difference, the larger the gain. When there is no
differentiation in the degenerating case ($K^{cp}=1$), the gain is
zero.
  \item If $K^{cp}> K^{sp}$, since the common part of two revenue $N_{eff}(K) \left(\frac{\bar{s}(K)\bar{\theta}(K)}{\bar{s}(K)+1}\right)= \frac{S\bar{\theta}N_{eff}(K)}{S+N_{eff}(K)}$ is a strictly increasing function of $N_{eff}(K)$, price differentiation makes more revenue even if the positive differentiation gain $g(K^{cp})$ is not taken into consideration. This result is intuitive: more consumers with purchasing power  always mean more revenue in the service provider's pocket.
\end{itemize}

Finally, we note that the $CP$ scheme in Section~\ref{sec:cpd} requires the complete network information. The $SP$ scheme here, on the other hand, works in the incomplete information case as well. This distinction becomes important in Section~\ref{sec:incomplete_information}.

\section{Partial Price Differentiation under Complete Information}
\label{sec:PPD}
For a service provider facing thousands of user types, it is often impractical to design a price choice for each user type. \rev{The reasons behind this, as discussed in \cite{valancius2011many}, are mainly high system overheads and customers' aversion.}
However, as we have shown in Sec.~\ref{sec:single_price}, the single pricing scheme may suffer a considerable revenue loss.  How to achieve a good  tradeoff  between the implementational complexity and the total revenue?
In reality, we usually see that the service provider \rev{offers} only a few pricing plans for the entire users population; we term it as the \emph{partial price differentiation} scheme. In this section, we will answer the following question: if the service provider is constrained to maintain a limited number of prices, $p^1,\dots,p^J$, $J\le I$, then what is the optimal pricing strategy and the maximum revenue?
Concretely, the Partial Price differentiation ($PP$) problem is formulated as follows.
\begin{eqnarray}
   PP: &\underset{n_i,p_i,s_i,p^j,a_i^{j}}{ {\text{maximize}}} & \sum\limits_{i\in\mathcal{I}}  n_i p_i s_i \nonumber\\
    &\text{subject to} & s_i =\left(\frac{\theta_i }{p_i } - 1\right)^+,\, \forall\, i\in \mathcal{I}, \label{eq:ppd_ic}\\
    && n_i \in\{0,\ldots,N_{i}\}, \; \forall\, i\in \mathcal{I}, \label{eq:ad}\\
    && \sum\limits_{i\in\mathcal{I}} n_i s_i \le S, \label{eq:resource}\\
    && p_i=\sum\limits_{j\in\mathcal{J}}a_{i}^{j}p^j, \label{eq:ppd_price}\\
    && \sum_{j\in\mathcal{J}}\!\!a_{i}^{j}=1,\, a_i^j\in\{0,1\}, \forall\, i\in \mathcal{I}. \label{eq:ppd_price_2}
\end{eqnarray}
Here $\mathcal{J}$ denotes the set $\{1,2,\dots,J\}$. \rev{Since we consider the complete information scenario in this section, the service provider can choose the price charged to each group, thus constraints (\ref{eq:ppd_ic}) -- (\ref{eq:resource}) are the same as in the $CP$ problem.
Constraints (\ref{eq:ppd_price}) and (\ref{eq:ppd_price_2}) mean that $p_i$ charged to each group~$i$ is one of $J$ choices from the set $\{p^j, j\in\mathcal{J}\}$.}  For convenience, we define \emph{cluster}  $\mathcal{C}^j\overset{\Delta}{=}\{i\,|\,a_i^j=1\},\; j\in \mathcal{J}$,  which is a set of groups charged with the same price $p^j$.  We use superscript $j$ to denote clusters, and subscript $i$ to denote groups through this section. We term the binary variables $\boldsymbol{a}\overset{\Delta}{=}\{a_i^j,\;j\in\mathcal{J},\; i\in\mathcal{I}\}$ as the \emph{partition}, which determines which cluster each group belongs to.

The $PP$ problem is a combinatorial optimization problem, and is more difficult than  the previous $CP$ and $SP$ problems.
On the other hand, we notice that the $PP$ problem formulation includes the $CP$ scheme ($J=I$) and the $SP$ scheme scenario ($J=1$) as special cases. The insights we obtained from solving these two special cases in Sections~\ref{sec:cpd} and \ref{sec:single_price} will be helpful to solve the general $PP$ problem.

\subsection{Three-level Decomposition}
To solve the $PP$ problem, we decompose and tackle it in three levels. In the lowest level-3, we determine the pricing and resource allocation for each cluster, given a fixed partition and fixed resource allocation among clusters. In level-2, we compute the optimal resource allocation among clusters, given a fixed partition. In level-1, we optimize the partition among groups.
\subsubsection{Level-3: Pricing and resource allocation in each cluster}
 For a fix partition $\boldsymbol{a}$ and a cluster resource allocation $\boldsymbol{s}\overset{\Delta}{=}\{s^j\}_{j\in\mathcal{J}}$, we focus the pricing and resource allocation
problems within each cluster $\mathcal{C}^j$, $j\in \mathcal{J}$:
\begin{eqnarray*}
   \text{Level-3:}
 & \underset{n_i,s_i,p^j}{{\text{maximize}}}\;\; &\sum_{i\in C^j} n_ip^js_i  \nonumber\\
    &\text{subject to}\;\; & s_i =\left(\frac{\theta_i }{p^j } - 1\right)^+,\quad \forall\, i\in \mathcal{C}^j, \\
    &&  n_i \le N_i, \quad \forall\, i\in \mathcal{C}^j,\\
    && \sum\limits_{i\in \mathcal{C}^j} n_i s_i \le s^j. 
\end{eqnarray*}
The level-3 subproblem coincides with the $SP$ scheme  discussed in Section~\ref{sec:single_price}, since all groups within the same cluster $\mathcal{C}^j$ are charged with a single price $p^j$. We can then directly apply the results in Theorem~\ref{alg:2} to solve the Level-3 problem. We denote the effective market threshold\footnote{Note that we do not assume that the effective market threshold equals to the number of effective groups, e.g., there are 2 effective groups in Fig.~5, but threshold $K^j=5$. Later we will prove that there is unified threshold for the $PP$ problem. Then by this result, the group index threshold actually coincides with the number of effective groups.} for cluster $\mathcal{C}^j$  as $K^j$, which can be computed in Algorithm~\ref{alg:2}. An illustrative  example is shown in Fig.~\ref{fig:cluster}, where the cluster contains four groups (group~4, 5, 6 and 7), and the effective market contains groups~4 and 5, thus $K^j=5$. The service provider obtains the following maximum revenue obtained from cluster $\mathcal{C}^j$:
\begin{equation}
\label{eq:ppd_level_3}
R^j(s^j,\boldsymbol{a})=\frac{s^j \sum_{i\in C^j, \,i\le K^j}N_i\theta_i}{s^j +\sum_{i\in C^j, \,i\le K^j}N_i}.
\end{equation}
\begin{figure}[ht]
\centering
\includegraphics[scale=0.48]{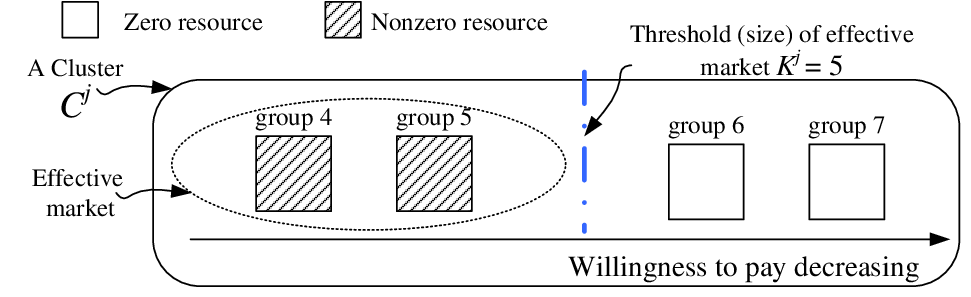}
\caption{An illustrative  example of Level-3: the cluster contains four groups, group~4, 5, 6 and 7; and the effective market contains group~4 and 5, thus $K^j=5$}
\label{fig:cluster}
\end{figure}

\subsubsection{Level-2: Resource allocation among clusters} For a fix partition $\boldsymbol{a}$, we then consider the resource allocation among clusters.
\begin{eqnarray*}
  \text{Level-2:} & \underset{ s^j\ge 0}{\text{maximize}}\;\; &\sum_{j\in \mathcal{J}} R^j(s^j,\boldsymbol{a})  \nonumber\\
    &\text{subject to}\;\; & \sum\limits_{j\in \mathcal{J}} s^j \le S \label{con:resource_2}
\end{eqnarray*}
We will show in Section~\ref{sub_L2L3} that subproblems in Level-2 and Level-3 can be transformed into a complete price differentiation problem under proper technique conditions. Let us denote the its optimal value as $R_{pp}(\boldsymbol{a})$.

\subsubsection{Level-1: cluster partition} Finally, we solve the cluster partition problem.
\begin{eqnarray*}
   \text{Level-1:} & \underset{a_i^j\in\{0,1\}} {\text{maximize}}& R_{pp}(\boldsymbol{a})  \nonumber\\
    &\text{subject to} &\sum_{j\in\mathcal{J}}a_{i}^{j}=1,\; i\in \mathcal{I}.
\end{eqnarray*}
This partition problem is a combinatorial optimization problem.
The size of its feasible set is $S(I,J)=\frac{1}{J!}\sum\limits_{t=1}^{J}(-1)^{J+t}C(J,t)t^I$, \emph{Stirling number of the second kind}
 \cite[Chap.13]{van2001course}, where $C(J,t)$ is the binomial coefficient. Some numerical examples are given in the third row in Table~\ref{tab:searching_space}. If the number of prices $J$ is given, the feasible set size is exponential in the total number of groups $I$. For our problem, however, it is possible to reduce the size of the feasible set by exploiting the special problem structure. More specifically, the group indices in each cluster should be consecutive at the optimum. This means that the size of the feasible set is $C(I-1,J-1)$  as shown in the last row in Table~\ref{tab:searching_space}, and thus is much smaller than $S(I,J)$.

\begin{table*}
\caption{Numerical examples for feasible set size of the partition problem in Level-1}
\label{tab:searching_space} \centering
\begin{tabular}{c|c|c|c|c|c}
\hline
Number of groups &\multicolumn{2}{|c|}{$I=10$}&\multicolumn{2}{|c|}{$I=100$}&{$I=1000$}\\
\hline
Number of prices &$J=2$ & $J=3$&$J=2$&$J=3$&$J=2$\\
\hline
$S(I,J)$&$511$&$9330$&$6.33825\times 10^{29}$&$8.58963\times 10^{46}$&$5.35754\times 10^{300}$\\
\hline
$C(I-1,J-1)$&$9$&$36$&$99$&$4851$&$999$\\
\hline
\end{tabular}
\end{table*}

Next we discuss how to solve the three level subproblems. A route map for the whole solving process is given in Fig.~\ref{fig:PPD}.

\begin{figure}[ht]
\centering
\includegraphics[scale=0.52]{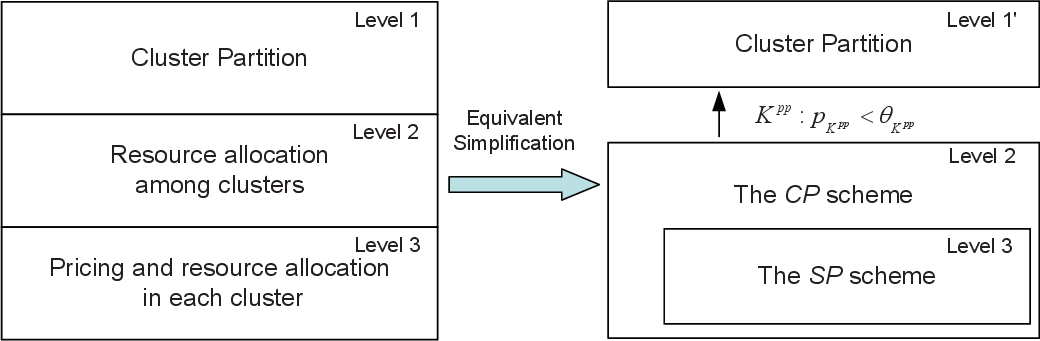}
\caption{Decomposition and simplification of the general $PP$ problem: The three-level decomposition structure of the $PP$ problem is shown in the left hand side. After simplifications in Section~\ref{sub_L2L3} and \ref{sub:L1}, the problem will be reduced to structure in right hand side.}
\label{fig:PPD}
\end{figure}

\subsection{Solving Level-2 and Level-3}
\label{sub_L2L3}
The optimal solution  (\ref{eq:ppd_level_3}) of the Level-3 problem  can be equivalently written as
\begin{equation}
\label{eq:ppd_l3}
R^j(\boldsymbol{s},\boldsymbol{a})=\frac{s^j \sum_{i\in C^j, \,i\le K^j}N_i\theta_i}{s^j +\sum_{i\in C^j, \,i\le K^j}N_i}\overset{(a)}{=} \frac{s^j N^j\theta^j}{s^j +N^j},
\end{equation}
\begin{equation}
\label{eq:ppd_level_3_condition}
\text{where}\;\;
\begin{cases}
N^j&={\sum_{i\in C^j, \,i\le K^j}N_i},\\
\theta^j&=\sum_{i\in C^j, \,i\le K^j}\frac{N_i \theta_i}{N^j}.
\end{cases}
\end{equation}
The equality (a) in (\ref{eq:ppd_l3}) means that each cluster  $\mathcal{C}^j$ can be equivalently treated as a group with $N^j$ homogeneous users with the same willings to pay $\theta^j$. We name this equivalent group as a \emph{super-group} (SG). We summarize the above result as the following lemma.

\begin{lemma}
\label{le:sp}
For every cluster $C^j$ and total resource $s^j$, $j\in\mathcal{J}$, we can find an equivalent super-group which satisfies conditions in (\ref{eq:ppd_level_3_condition}) and achieves the same revenue under the $SP$ scheme.
\end{lemma}

Based on Lemma~\ref{le:sp}, the level-2 and level-3 subproblems together can be viewed as the $CP$ problem for super-groups. Since a cluster and its super-group from a one-to-one mapping, we will use the two words interchangeably in the sequel.

However, simply combining Theorems~\ref{thm:opt} and \ref{thm:sp} to solve the level-2 and level-3 subproblems for a fixed partition $\boldsymbol{a}$ can result in a very high complexity.  This is because the effective markets within each super-group and between super-groups are coupled together. An illustrative  example of this coupling effective market is shown in Fig.~\ref{fig:two_level_threshold}, where $K^c$ is the threshold between clusters and has three possible positions (i.e., between group~2 and group~3, between group~5 and group~6, or after group~6); and $K_1$ and $K_2$ are thresholds within cluster $\mathcal{C}^1$ and  $\mathcal{C}^2$, which have two or three possible positions, respectively. Thus, there are $(2\times 3)\times 3=18$ possible thresholds possibilities in total.

\begin{figure}[ht]
\centering
\includegraphics[scale=0.41]{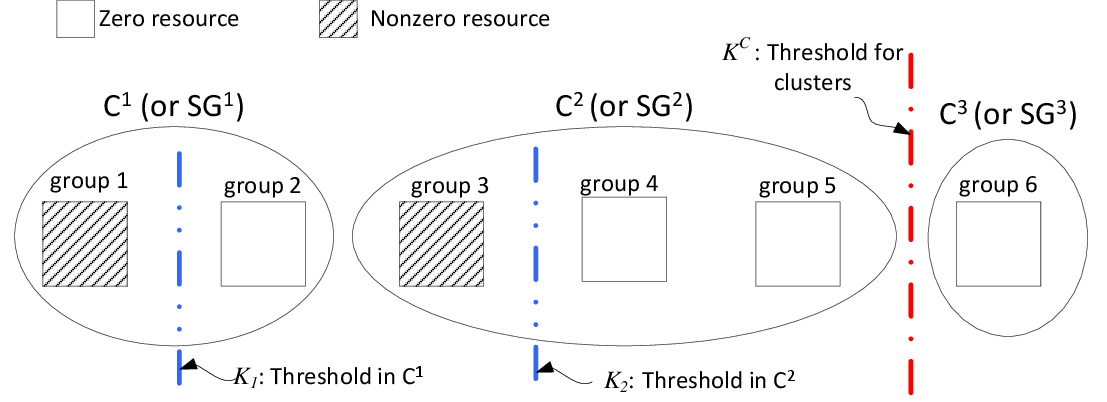}
\caption{An example of coupling thresholds. }
\label{fig:two_level_threshold}
\end{figure}

The key idea to resolve this coupling issue is to show that the situation in Fig.~\ref{fig:two_level_threshold} can not be an optimal solution of the $PP$ problem. The results in Sections~\ref{sec:cpd} and \ref{sec:single_price} show that there is a unified threshold at the optimum in both the $CP$ and $SP$ cases, e.g., Fig.~\ref{fig:effective_market_concept}.  Next we will show that a unified single threshold also exists in the $PP$ case.

\begin{lemma}
\label{le:consecutive_index} At any optimal solution of the $PP$ scheme, the group indices of the effective market is consecutive.
\end{lemma}

\rev{The proof of Lemma~\ref{le:consecutive_index} can be found in Appendix~\ref{proof_le:consectutive_index}.} The intuition is that the resource should be always allocated to high willingness to pay users at the optimum. Thus, it is not possible to have Fig.~\ref{fig:two_level_threshold}  at an optimal solution, where high willingness to pay users in group~2 are allocated zero resource while low willingness to pay users in group~3 are allocated positive resources.

Based on Lemma~\ref{le:consecutive_index}, we know that there is a unified effective market threshold for the $PP$ problem, denoted as $K^{pp}$. Since all groups with indices larger than $K^{pp}$ make zero contribution to the revenue,  we can ignore them and only consider the partition problem for the first $K^{pp}$ groups. Given a partition that divides the $K^{pp}$ groups into $J$ clusters (super-groups), we can apply the $CP$ result in Section~\ref{sec:cpd} to compute the optimal revenue in Level-2 based on Theorem~\ref{thm:opt}.
\begin{align}
R_{pp}(\boldsymbol{a})
=\sum_{j=1}^{J}{N^{j}}{\theta^j}-\frac{\left(\sum_{j=1}^{J}{N^j}\sqrt{{\theta^j}}\right)^2} {S+\sum_{j=1}^{J}{N^j}}\nonumber \\
=\sum_{i=1}^{K^{pp}}N_i\theta_i-\frac{\left(\sum_{j=1}^{J}{N^j}\sqrt{{\theta^j}}\right)^2}
{S+\sum_{i=1}^{ K^{pp}}N_i}.\label{eq:R_ppd_K0}
\end{align}

\subsection{Solving Level-1}
\label{sub:L1}
\subsubsection{With a given effective market threshold $K^{pp}$}
Based on the previous results, we first simplify the level-1 subproblem, and prove the theorem below. 

\begin{theorem}
For a given threshold $K^{pp}$, the optimal partition of Level-1 is the solution of the following optimization problem.
\begin{eqnarray}
\text{Level-1}'&\underset{a_i^j, N^j, \theta^j}{\text{minimize}}& \sum_{j\in \mathcal{J}} N^j\sqrt{\theta^j}\label{eq:obj}\nonumber\\
&\text{subject to}& N^j=\sum_{i\in\mathcal{K}^{pp}}N_i a_i^j,\;\;j\in\mathcal{J},\nonumber\\
&& \theta^j=\sum_{i\in\mathcal{K}^{pp}}\frac{N_ia_i^j}{N^j}\theta_i\;\;j\in\mathcal{J},\label{eq:theta}\nonumber\\
&&\sum_{j\in \mathcal{J}}a_i^j=1,\;a_i^j\in\{0,1\}\;,i\in\mathcal{K}^{pp}\;j\!\in\!\mathcal{J},\nonumber\\
&& \theta_{K^{pp}}>p^J=\sqrt{\theta^J(\boldsymbol{a})\lambda(\boldsymbol{a})}.\label{eq:pd_threshold_con}
\end{eqnarray}
where
$\mathcal{K}^{pp}\overset{\Delta}{=}\{1,2,\dots,K^{pp}\}$,  $\theta^J(\boldsymbol{a})$ is the value of average willingness to pay of the $J$th group for the partition $\boldsymbol{a}$, and $\lambda(\boldsymbol{a})=\left(\frac{\sum_{j\in\mathcal{J}}{N^j}\sqrt{{\theta^j}}}{S+\sum_{i=1}^{ K^{pp}}N_i}\right)^2.$
\end{theorem}

\begin{proof}
The objective function and the first three constraints in Level-1$'$ are easy to understand:
if the effective market threshold $K^{pp}$ is given, then the objective function of the Level-1 subproblem, maximizing $R_{pp}$ in (\ref{eq:R_ppd_K0}) over $\boldsymbol{a}$, is as simple as minimizing  $\sum_{j=1}^{J}{N^j}\sqrt{{\theta^j}}$ as the level-1$'$ problem suggested;
the first three constraints are given by the definition of the partition.

Constraint (\ref{eq:pd_threshold_con}) is the threshold condition that supports  (\ref{eq:R_ppd_K0}), which means that the least willingness to pay users in the effective market has a positive demand.  It ensures that calculating the revenue by (\ref{eq:R_ppd_K0}) is valid. Remember that the solution of the $CP$ problem of Level-2 and Level-3 is threshold based, and Lemma~\ref{le:consecutive_index} indicates that (\ref{eq:pd_threshold_con}) is sufficient for that all groups with willingness larger than group $K^{pp}$ can have positive demands. Otherwise, we can construct another partition leading to a larger revenue (please refer to the proof of Lemma~\ref{le:consecutive_index}), or equivalently leading to a less objective value of Level-1$'$. This leads to a contradiction.
\end{proof}

The level-1$'$ problem is still a combinatorial optimization problem with a large feasible set of $\boldsymbol{a}$ (similar as the original Level-1). The following result can help us to reduce the size of the feasible set.
\begin{theorem}
\label{th:consecutive}
For any effective market size $K^{pp}$ and number of prices $J$,  an optimal partition of the $PP$ problem involves consecutive group indices within clusters.
\end{theorem}

\rev{The proof of Theorem~\ref{th:consecutive} is given in Appendix~\ref{proof_th:consecutive}.} We first prove this result is true for Level-1$'$ without constraint (\ref{eq:pd_threshold_con}), and further show that this result will not affected by (\ref{eq:pd_threshold_con}). The intuition is that high willingness to pay users should be allocated positive resources with priority. It implies that groups with similar willingness to pays should be partitioned in the same cluster, instead of in several far away clusters. Or equivalently, the group indices within each cluster should be consecutive.

We define $\mathcal{A}$ as the set of all partitions with consecutive group indices within each cluster, and  $v(\boldsymbol{a})=\sum_{j\in\mathcal{J}}N^j\sqrt{\theta^j}
$ is the value of objective of Level-1$'$ for a partition $\boldsymbol{a}$. Algorithm ~\ref{al:pd_simplifiedlevel-1} finds the optimal solution of  Level-1$'$. The main idea for this algorithm is to enumerate every possible partition in set $\mathcal{A}$, and then check whether the threshold condition (\ref{eq:pd_threshold_con}) can be satisfied. The main part of this algorithm is to enumerate all partitions in set $\mathcal{A}$ of $C(K^{pp}-1, J-1)$ feasible partitions. Thus the
complexity of Algorithm~\ref{al:pd_simplifiedlevel-1} is no more than $\mathcal{O}((K^{pp})^{J-1})$.
\begin{algorithm}[htb]
\caption{Solve the level-1$'$ problem with fixed $K^{pp}$}
\label{al:pd_simplifiedlevel-1}
\begin{algorithmic}[1]
    \Function{Level-1}{$K^{pp}$, $J$}
    \State $k\gets K^{pp}$
    \State $v^*\gets \sqrt{\sum_{i=1}^{k}N_i\sum_{i=1}^{k}N_i\theta_i}$, $\boldsymbol{a}^*=\boldsymbol{0}$ 
    \For{$\boldsymbol{a}\in\mathcal{A}$}
			\If{$\theta_{k}> \sqrt{\theta^J(\boldsymbol{a})\lambda(\boldsymbol{a})}$} 
				\If{$v(\boldsymbol{a})<v^*$}
					\State $v^*\gets v(\boldsymbol{a})$, $\boldsymbol{a}^*\gets\boldsymbol{a}$
				\EndIf
			\EndIf
    \EndFor
    \State \textbf{return} $\boldsymbol{a}^*$
    \EndFunction
\end{algorithmic}
\end{algorithm}

\subsubsection{Search the optimal effective market threshold  $K^{pp}$}
We know the optimal market threshold $K^{pp}$ is upper-bounded, i.e., $K^{pp}\le K^{cp}\le I$. Thus we can first run Algorithm~{\ref{alg:1}} to calculate the effective market size for the $CP$ scheme $K^{cp}$. Then, we search the optimal $K^{pp}$ iteratively  using Algorithm~\ref{al:pd_simplifiedlevel-1} as an inner loop. We start by letting  $K^{pp}=K^{cp}$ and run Algorithm~\ref{al:pd_simplifiedlevel-1}. If there is no solution, we decrease $K^{pp}$ by one and run Algorithm~\ref{al:pd_simplifiedlevel-1} again. The algorithm will terminate once we find an effective market threshold where Algorithm~\ref{al:pd_simplifiedlevel-1} has an optimal solution.
Once the optimal threshold and the partition of the clusters are determined, we can further run  Algorithm~{\ref{alg:1}} to solve the joint optimal resource allocation and pricing scheme. The pseudo code is given in Algorithm~\ref{alg:3} as follows.

\begin{algorithm}[htb]                      
\caption{Solve Partial Price Differentiation Problem}     
\begin{algorithmic}[1]
		\State $p_i\gets\theta_i$ 
    \State $k\Leftarrow$ {CP}($\{N_i,\theta_i\}_{i\in{\mathcal{I}}}$, $S$)$\_1$, $\boldsymbol{a}^*\Leftarrow$ Level-1($k$,$J$)
		\While{$ \boldsymbol{a}^\ast = =\boldsymbol{0}$} 
			\State $k\gets k-1$, $\boldsymbol{a}^*\Leftarrow$ Level-1($k$,$J$)
		\EndWhile  
		\For{$j\gets 1, J$} 
\State $N^j\gets\sum_{i=1}^k N_i a_i^j$, $\theta^j\gets\sum_{i=1}^k\frac{N_ia_i^j}{N^j}\theta_i$
		\EndFor
		\State $\lambda\Leftarrow$ CP($\{N^j,\theta^j\}_{i\in{\mathcal{J}}}$, $S$)$\_2$
	\For {$i\gets 1,k$} 
		\State $p_i\gets \sum_{j=1}^{J}a_i^j\sqrt{\theta^j\lambda}$
	\EndFor
   \State \textbf{return} $\{p_i\}_{i\in\mathcal{I}}$
\end{algorithmic}
\label{alg:3}                      
\end{algorithm}
In Algorithm~\ref{alg:3}, it invokes two functions: CP($\{N_i\theta_i\}_{i\in{\mathcal{I}}}$, $S$) as described in Algorithm~\ref{alg:1} and  and Level-1($k,J$) as in Algorithm~\ref{al:pd_simplifiedlevel-1}. CP($\{N_i\theta_i\}_{i\in{\mathcal{I}}}$, $S$) returns a vector with two elements: CP($\{N_i\{\theta_i\}_{i\in{\mathcal{I}}}$, $S$)$\_1$ denotes the first element $K^{cp}$, and {CP}($\{N_i\theta_i\}_{i\in{\mathcal{I}}}$, $S$)$\_2$ denotes the second element $\lambda^*$ in the $CP$ problem.

The above analysis leads to the following theorem:
\begin{theorem}
The solution obtained by Algorithm~\ref{alg:3} is optimal for the $PP$ problem.
\end{theorem}

\begin{proof}
It is clear that Algorithm~\ref{alg:3} enumerates every possible value of the effective market size for the $PP$ problem $K^{pp}$, and for a given $K^{pp}$ its inner loop Algorithm~\ref{al:pd_simplifiedlevel-1} enumerates every possible partition in set $\mathcal{A}$. Therefore, the result in Theorem~\ref{alg:3} follows.
\end{proof}

Next we discuss the complexity of Algorithm~\ref{alg:3}. The complexity of Algorithm~\ref{alg:1} is $\mathcal{O}(I)$, and we run it twice in Algorithm~\ref{alg:3}. The worst case complexity of  Algorithm~\ref{al:pd_simplifiedlevel-1}  is $\mathcal{O}({I}^{J-1})$, and we run it no more than $I-J$ times. Thus the whole complexity of Algorithm~\ref{alg:3} is no more than $\mathcal{O}({I}^{J})$, which is polynomial of $I$.

\section{Price Differentiation under Incomplete Information}
\label{sec:incomplete_information}
In Sections~\ref{sec:cpd}, \ref{sec:single_price}, and \ref{sec:PPD},
we discuss various pricing schemes with different implementational complexity level under complete information, the revenues of  which can be viewed as the benchmark of practical pricing designs. In this section, we further study the incomplete information scenario, where the \rev{service provider} does not know the group association of each user. The challenge for pricing in this case is that the \rev{service provider} needs to provide the right incentive so that a group $i$ user does not want to pretend to be a user in a different group. It is clear that the $CP$ scheme in Section~\ref{sec:cpd} and the $PP$ scheme in Section~\ref{sec:PPD} cannot be directly applied here. The $SP$ scheme in Section~\ref{sec:single_price} is a special case, since it does not require the user-group association information in the first place and thus can be applied in the incomplete information scenario directly. On the other hand, we know that the $SP$ scheme may suffer a considerable revenue loss compared with the $CP$ scheme. Thus it is natural to ask whether it is possible to design an incentive compatible differentiation scheme under incomplete information.
In this section, we design a quantity-based price menu to incentivize users to make the right self-selection and to achieve the same maximum revenue of the $CP$ scheme under complete information with proper technical conditions.  We name it as the Incentive Compatible Complete Price differentiation ($ICCP$) scheme.

In the $ICCP$ scheme, the service provider publishes a quantity-based price menu, which consists of several step functions of resource \rev{quantities}. Users are allowed to freely choose their quantities. The aim of this price menu is to make the users \emph{self-differentiated}, so that to mimic the same result (the same prices and resource allocations) of the $CP$ scheme under complete information.
Based on Theorem \ref{thm:opt}, there are only $K$ (without confusion, we remove the superscript ``cp'' to simplify the notation) effective groups of users receiving non-zero resource allocations, thus there are $K$ steps of unit prices, $p^*_{1}>p^*_{2}>\cdots
>p^*_{K}$ in the price menu. These prices are exactly the same optimal prices that the \rev{service provider} would charge for $K$ effective groups as in Theorem \ref{thm:opt}. Note that for the $K+1,\dots,I$ groups, all the prices in the menu are too high for them, then they will still demand zero resource.
The quantity is divided into $K$ intervals by $K-1$ thresholds, $s_{th}^{1}>s_{th}^{2}>\cdots>s_{th}^{K-1}$. The $ICCP$ scheme can specified as
follows:
\begin{equation}
p(s)= \left\{ {\begin{array}{ll}
   {p^*_{1}\quad\;\;\; {\rm{when}}\; s > s^1_{th}   }\\
   {p^*_2\quad\;\;\; {\rm{when}}\;  s^1_{th}  \ge s > s^2_{th}  } \\
    \vdots   \\
   {p^*_{K}\;\;\quad {\rm{when}}\; s^{K-1}_{th}  \ge s > 0.  }\\
\end{array} }\right.
\label{IC_P}
\end{equation}

A four-group example is shown in Fig.~\ref{fig:IC_pricing}.
\begin{figure}[ht]
\centering
\includegraphics[scale=0.65]{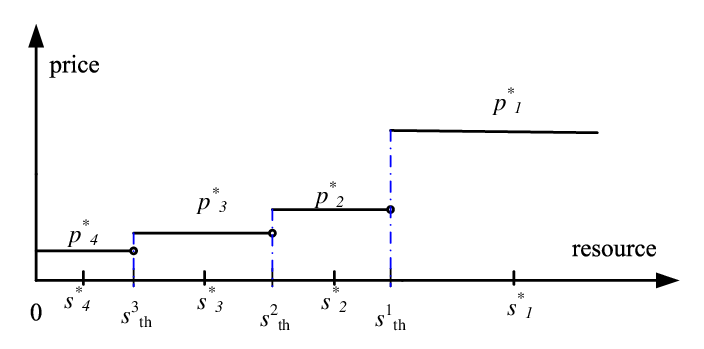}
\caption{A four-group example of the $ICCP$ scheme: where the prices
$p^*_1>p^*_2>p^*_3>p^*_{4}$ are the same as the $CP$ scheme. To mimic the same resource
allocation as under the $CP$ scheme, one necessary  (but not sufficient) condition is $s^{j-1}_{th}\ge s^*_{j}$ for all $j$, where $s^*_j$ is
the optimal resource allocation of the $CP$ scheme.}
\label{fig:IC_pricing}
\end{figure}

Note that in \rev{contrast} to the usual ``volume discount'', here the price is non-decreasing in quantity. This is motivated by the resource allocation in Theorem \ref{thm:opt}, that a user with a higher $\theta_{i}$ is charged a higher price for a larger resource allocation. Thus the observable quantity can be viewed as an indication of the unobservable users' willingness to pay, and help to realize price differentiation under incomplete information.

The key challenge in the $ICCP$ scheme is to properly set the quantity thresholds so
that users are perfectly segmented through self-differentiation.
This is, however, not always possible. Next we derive the necessary
and sufficient conditions to guarantee the perfect segmentation.

Let us first study the self-selection problem between two groups:
group $i$ and group $q$ with $i<q$. Later on we will generalize the
results to multiple groups. Here group $i$ has a higher willingness to pay, and will be charged with a higher price $p_i^\ast$ in the $CP$ case.   The incentive compatible constraint is that a high willingness to pay user can not get more surplus by pretending to be a low willingness to pay user, i.e.,
$\max\limits_{s}U_i(s;p^*_i)\ge \max\limits_{s}U_i(s;p^*_q),$
where $U_i(s;p)=\theta_i\ln(1+s)-ps$ is the surplus of
a group $i$ user when it is charged with price $p$.

Without confusion, we still use $s^*_i$ to denote the optimal resource allocation
under the optimal prices in Theorem \ref{thm:opt}, i.e.,
$s^*_i=\arg\max\limits_{s_i\geq 0}U_i(s_i;p^*_i).$ We define $s_{i\rightarrow q}$ as the quantity satisfying
\begin{equation}
\label{eq:trans}
\left\{\!\!{\begin{array}{ll}
 \!\!  U_i(s_{i\rightarrow q};p^*_q)= U_i(s^*_i;p^*_i)\!\!\\
 \!\!  s_{i\rightarrow q}< s^*_i
\end{array}} \right.
. \end{equation}
In other words, when a group $i$ user is charged with a lower price
$p^*_q$ and demands resource quantity $s_{i\rightarrow q}$, it
achieves the same as the maximum surplus under the optimal price of the $CP$ scheme $p^*_{i}$,
as showed in Fig.~\ref{fig:utility}.  Since the there two solutions of the first equation of (\ref{eq:trans}), we constraint $s_{i\rightarrow q}$ to be the one that is smaller than $s_i^\ast$.
\begin{figure}[ht]
 \centering
\includegraphics[scale=0.45]{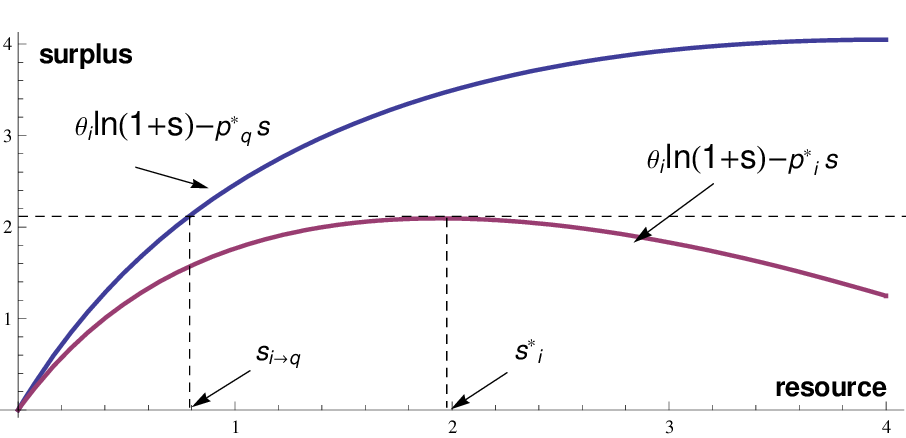}
\caption{When the threshold $s^{q-1}_{th}<s_{i\rightarrow q}$, the
group $i$ user can not obtain $U(s_i^*,p^*_i)$ if it chooses the
lower price $p_q$ at a quantity less than $s^{q-1}_{th}$. Therefore
it will automatically choose the high price $p^*_i$ to maximize its
surplus. } \label{fig:utility}
\end{figure}

To maintain the group $i$ users' incentive to choose the higher
price $p^*_i$ instead of $p^*_q$, we must have $s^{q-1}_{th}\le
s_{i\rightarrow q}$, which means a group $i$ user can not obtain
$U_i(s_i^*,p^*_{i})$ if it chooses a quantity less than
$s^{q-1}_{th}$. In other words, it will automatically choose the
higher (and the desirable) price $p^*_i$ to maximize its surplus.
On the other hand, we must have $s_{th}^{q-1} \geq s_q^\ast$ in order to maintain the optimal resource allocation and allow a group $q$ user to choose the right quantity-price combination (illustrated in Fig.~\ref{fig:IC_pricing}).

Therefore, it is clear that the \emph{necessary and sufficient}
condition that the $ICCP$ scheme under incomplete information achieves the same maximum revenue of the $CP$ scheme
under complete information is
\begin{equation}
    s^*_q\le s_{i\rightarrow q}, \;\forall\;i<q, \forall\; q\in \{2,\dots,K\}.
\label{IC2}
\end{equation}
By solving these inequalities, we can obtain the following theorem
(\rev{detailed proof in Appendix~\ref{sub_appendix_IC}}).

\begin{theorem}
\label{the:IC}
There exist unique thresholds $\{t_{1},\ldots,\!t_{K\!-1}\!\}$, such that the $ICCP$ scheme achieves
the same maximum revenue as in the complete information case if
$$\sqrt{\frac{\theta_{q}}{\theta_{q+1}}}\geq t_{q}\quad\text{for}\;q\!=\!1,\dots,K-1.$$
 Moreover, $t_{q}$ is the unique solution of
the equation
$$t^2\ln t-(t^2-1)+
\frac{t\sum_{k=1}^{q}N_k+N_{q+1}}{S+\sum_{k=1}^{K^{cp}}N_k}(t-1)=0$$
over the domain $t>1$.
\end{theorem}

We want to mention that the condition in Theorem \ref{the:IC} is
necessary and sufficient for the case of $K=2$ effective
groups\footnote{There might be other groups who are not allocated
positive resource under the optimal pricing.}. For $K>2$, Theorem
\ref{the:IC} is sufficient but not necessary. The intuition of Theorem~\ref{the:IC} is that users need to be sufficiently different to achieve the maximum revenue.

The following result immediately follows Theorem \ref{the:IC}.
\begin{corollary}
\label{cor:icn}
The $t_q$s in Theorem \ref{the:IC} satisfy
$t_q<t_{root}$ for $q=1,\ldots,K-1$, where $t_{root}\approx
2.21846$ is the larger root of equation $t^2\ln t-(t^2-1)=0$.
\end{corollary}

The Corollary~\ref{cor:icn} means that the users do not need to be extremely different to achieve the maximum revenue.

When the conditions in Theorem \ref{the:IC} are not satisfied, there
may be revenue loss by using the pricing menu in (\ref{IC_P}). Since it is difficult to explicitly solve the parameterized transcend equation
(\ref{eq:trans}), we are not able to characterize the loss in a closed form yet.

\rev{
\subsection{Extensions to Partial Price Differentiation under Incomplete Information}
For any given system parameters, we can numerically check whether a partial price differentiation scheme can achieve the same maximum revenue under  both the complete and incomplete information scenarios. The idea is similar as we described in this section. Since the $PP$ problem can be viewed as the $CP$ problem for all effective super-groups, then we can check the $ICCP$ bound in Theorem~\ref{the:IC} for super-groups (once the super-group partition is determined by the searching using Algorithm~\ref{alg:3}).
Deriving an analytical sufficient condition (as in Theorem~\ref{the:IC}) for an incentive compatible partial price differentiation scheme, however, is highly non-trivial and is part of our future study.}

\section{Connections with the Classical Price Differentiation Taxonomy}

In economics, price differentiation is often categorized by the first/second/third degree price differentiation taxonomy\cite{pashigian1995price}. This taxonomy is often used in the context of unlimited resources and general pricing functions. The proposed schemes in our paper have several key differences from these standard concepts, mainly due to the assumption of limited total resources and the choice of linear usage-based pricing.

In the first-degree price differentiation, each user is charged a price based on its willingness to pay. Such a scheme is also called the perfect price differentiation, as it captures users' entire surpluses (\ie leaving users with zero payoffs). For the complete price differentiation scheme under complete information in Section~\ref{sec:cpd}, the service provider does not extract all surpluses from users, mainly due to the choice of linear price functions. All effective users obtain positive payoffs.

In the second-degree price differentiation, prices are set according to quantities sold (\eg the volume discount). The pricing scheme under incomplete information in Section~\ref{sec:incomplete_information} has a similar flavor of quantity-based charging. However, our proposed pricing scheme charges a \emph{higher} unit price for a \emph{larger} quantity purchase, which is opposite to the usual practice of volume discount. This is due to our motivation of mimicking the optimal pricing differentiation scheme under the complete information. Our focus is to characterize the sufficient conditions, under which the revenue loss due to incomplete information (also called ``information rent'' \cite{mussa1978monopoly, stokey1979intertemporal, maskin1984monopoly,mas1995microeconomic}) is zero.

In the third-degree price differentiation, prices are set according to some customer segmentation. The segmentation is usually made based on users' certain attributes such as  ages, occupations, and genders. The partial price differentiation scheme in Section~\ref{sec:PPD} is analogous to the third-degree price differentiation, but here the user segmentation is still based on users' willingness to pay. The motivation of our scheme is to reduce the implementational complexity.

\section{Numerical Results}
We provide numerical examples to quantitatively study several key properties of
price differentiation strategies in this section.
\subsection{When is price differentiation most beneficial?}
\begin{definition}
(Revenue gain) We define the revenue gain $G$ of one pricing scheme as the ratio of the revenue difference (between this pricing scheme and the single pricing scheme) normalized by  the revenue of single pricing scheme.
\end{definition}

In this subsection, we will study the revenue gain of the $CP$ scheme, i.e.,
$G(\boldsymbol{N},\boldsymbol{\theta},S)\overset{\Delta}{=}\frac{R_{cp}-R_{sp}}{R_{sp}}$, where  $\boldsymbol{N}\overset{\Delta}{=}\{N_i, \forall i\in\mathcal{I}\}$ denotes the number of users in each groups, $\boldsymbol{\theta}\overset{\Delta}{=}\{\theta_i, \forall i\in\mathcal{I}\}$ denotes their willingness to pays, and  $S$ is the total resource. Notice that this gain is the maximum possible differentiation gain among all $PP$ schemes.

We first study a simple two-group case.
According to Theorems~\ref{thm:opt} and \ref{thm:sp}, the
revenue under the $SP$ scheme and the $CP$ scheme can be calculated as follows:
$$
R^{sp}=\left\{ {\begin{array}{ll}
   \frac{S(N_1\theta_1+N_2\theta_2)}{N_1+N_2+S}, & 1\le t<\sqrt{\frac{S+N_1}{N_1}};\\
  \frac{S N_1\theta_1}{N_1+S}, &  t\ge \sqrt{\frac{S+N_1}{N_1}};\\
\end{array}} \right.
$$
and
$$
R^{cp}=\left\{ {\begin{array}{ll}
   \frac{S(N_1\theta_1+N_2\theta_2)+N_1N_2(\sqrt{\theta_1}-\sqrt{\theta_2})^2}{N_1+N_2+S}, & 1\le t<\frac{S+N_1}{N_1};\\
  \frac{S N_1\theta_1}{N_1+S}& t\ge \frac{S+N_1}{N_1}.\\
\end{array}} \right.
$$
where $t=\sqrt{\frac{\theta_1}{\theta_2}}>1$.

The revenue gain will depend on five parameters, $S$, $N_1$, $\theta_1$, $N_2$ and $\theta_2$.
To simplify notations,  let $N=N_1+N_2$ be the total
number of the users, $\alpha=\frac{N_1}{N}$ the percentage of
group 1 users, and $\bar{s}=\frac{S}{N}$ the level of normalized
available resource. Thus the revenue gain can be expressed as
\begin{equation}
    G(t,\alpha,\bar{s})=
 \begin{cases}
   \frac{\alpha(1-\alpha)(t-1)^2}{\bar{s}(1+\alpha(t^2-1))} & 1<t<\sqrt{\frac{\bar{s}+\alpha}{\alpha}}, \\
   \frac{(1-\alpha)(\bar{s}+\alpha-t\alpha)^2}{\alpha\bar{s}(1+\bar{s})t^2} & \sqrt{\frac{\bar{s}+\alpha}{\alpha}}\le t\le\frac{\bar{s}+\alpha}{\alpha}.
  \end{cases}
\end{equation}
Next we discuss the impact of each parameter.
\begin{observation}
In terms of the parameter $t$, $G$ monotonically increases in  $\left(1,\sqrt{\frac{\bar{s}+\alpha}{\alpha}}\right)$ and decrease in $\left(\sqrt{\frac{k+\alpha}{\alpha}},\frac{k+\alpha}{\alpha}\right)$. The maximum is obtained at $t_{G-max}=\sqrt{\frac{k+\alpha}{\alpha}}$, when the resource allocated to the group 2 user just becomes zero in the $SP$ scheme.
\end{observation}

One example is showed in Fig.\ref{fig:G_t}.
\begin{figure}
\centering
\includegraphics[scale=0.45]{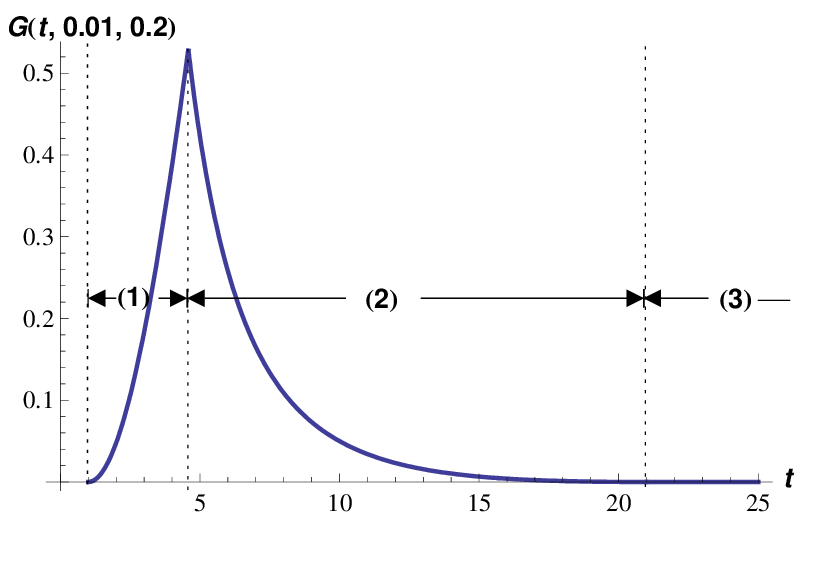}
\caption{One example of the revenue gain $G(t,0.01,0.2)$  for the $CP$ scheme. It is clear that the revenue gain can be divided into three regions.  Region(1), increasing region, where $K^{cp}=K^{sp}=2$, and the revenue gain comes from the differentiation gain. Region(2), decreasing region,  where $K^{cp}=2$, $K^{sp}=1$, and the revenue gain comes from larger effective market and differentiation gain. Region(3), zero region, where $K^{cp}=K^{sp}=1$, and is a degenerating case where two pricing scheme coincide. }
\label{fig:G_t}
\end{figure}

It is clear that the revenue gain is not monotonic in the willingness to pay ratio. Its behavior can be divided into three regions: the increasing Region (1) with $t\in \left(1,\sqrt{\frac{\bar{s}+\alpha}{\alpha}}\right)$, the decreasing Region (2) with $t\in\left[\sqrt{\frac{k+\alpha}{\alpha}},\frac{k+\alpha}{\alpha}\right)$, and the zero Region (3) with $t\ge\frac{k+\alpha}{\alpha}$.

It is also interesting to note that three regions  are closed related to the effective market sizes: $K^{sp}=K^{cp}=2$ in  Region (1);  $K^{sp}=1$ and $K^{cp}=2$ in Region (2); and  $K^{cp}=K^{sp}=1$ in Region (3) where the $CP$ scheme degenerates to the $SP$ scheme.  The peak point of the revenue gain correspond to the place where the effective market of the $SP$ Scheme changes.

Intuitively, the $CP$ scheme increases the revenue by charging the high willingness groups with high prices, thus the revenue gain increases first when the difference of willingness to pays increase. However, when the difference of willingness to pay is very large, the $CP$ scheme obtain most revenue from the high willingness to pay users, while the $SP$ scheme declines the low willingness to pay users but serves the high willingness to pays only. Both schemes lead to similar resource allocation in this region, and thus the revenue gain decreases as the difference of willingness to pays increases.

Figure~\ref{fig:G_t} shows the revenue gain under usage-based
pricing can be very high in some scenario, e.g., over $50\%$ in this
example. We can define this peak revenue gain as
\begin{equation*}
G_{max}(\alpha,\bar{s})\!=\!\max_{t\ge 1}
G(t,\alpha,\bar{s})\!
=\!\frac{\!(\alpha\!-\!1)(\!\sqrt{\bar{s}\!+\!\alpha}\!-\!\sqrt{\alpha})^2}{\bar{s}(1+\bar{s})\!}.
\end{equation*}
Figure~\ref{fig:gmax} is shown how $G_{max}$ changes in $\bar{s}$ with
different parameters $\alpha$.
\begin{figure}[htb]
\centering
\includegraphics[scale=0.38]{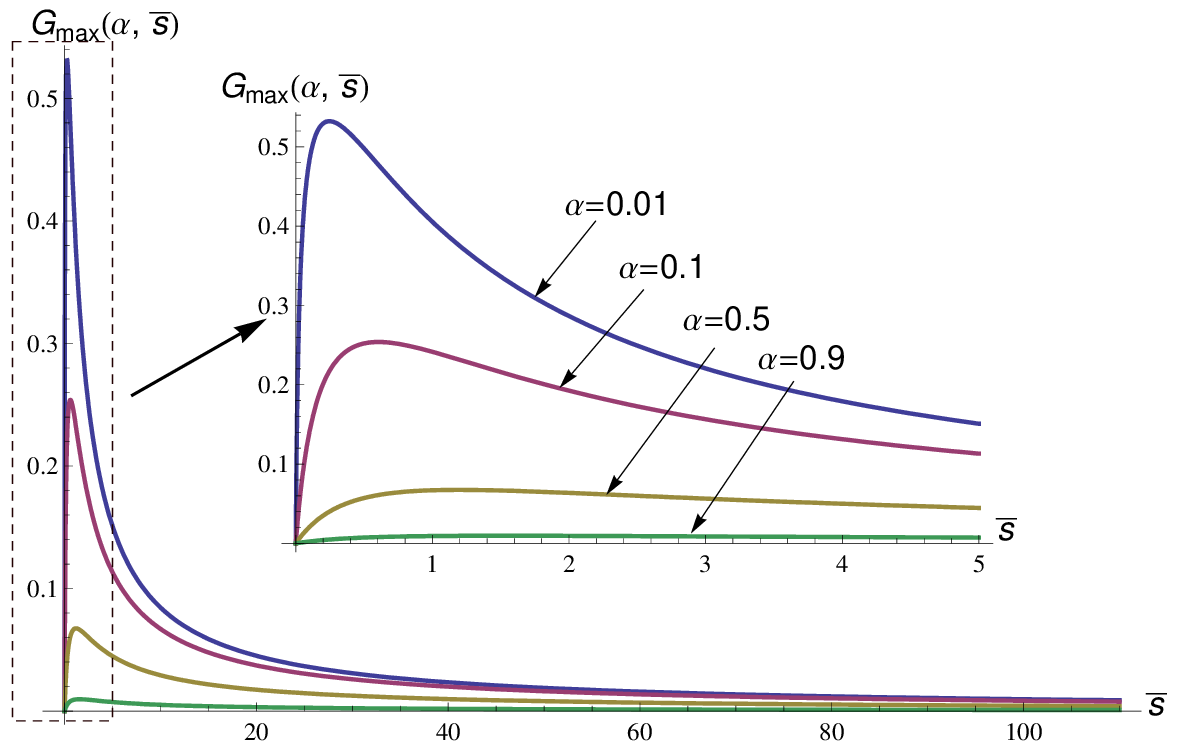}
\caption{For a fixed $\bar{s}$, $G_{max}(\alpha,\bar{s})$ monotonically increases in $\alpha$. For a fixed $\alpha$, $G_{max}(\alpha,\bar{s})$ first increases in $\bar{s}$, and then decreases in $\bar{s}$.} \label{fig:gmax}
\end{figure}
\begin{observation}
For a fixed $\bar{s}$, $G_{max}(\alpha,\bar{s})$ monotonically decreases in $\alpha$ .
\end{observation}

When $\alpha$ is small, which means high willingness to pay users are minorities in the effective market, the advantage of price differentiation is very evident. As shown in Fig.~\ref{fig:gmax}, when $\alpha= 0.1$, the maximum possible revenue gain can be over than $20\%$; and when $\alpha=0.01$, this gain can be even higher than $50\%$. However, when high willingness to pay users are majority, the price differentiation gain is very limited, for example,  the gain is no larger than $8\%$ and $2\%$  for $\alpha=0.5$ and 0.9, respectively.

Intuitively, high willingness to pay users are the most profitable users in the market. Ignoring them  is detrimental in terms of revenue even if they only occupy a small fraction of the population. Since the $SP$ scheme is set based on the average willingness to pay of the effective market, the high willingness to pay users will be ignored  (in the sense of not charging the desirable high price) when $\alpha$ is small. In contrast, ignoring the low willingness to pay users when $\alpha$ is large is not a big issue.

\begin{observation}
For parameter $k$, $G_{max}(\alpha,\bar{s})$ is not a monotonic function in $\bar{s}$. Its shape looks like a skewed bell. The gain is either small when $\bar{s}$ is very small or very large.
\end{observation}

Small $\bar{s}$ means that resource is very limited, and both schemes allocates the resource to high willingness to pay users (see the discussion of the threshold structure in Sections~\ref{sec:cpd} and \ref{sec:single_price}), and thus there is not much difference between two pricing schemes.
While $\bar{s}$ is very large, i.e., the resource is abundant, the prices and the resource allocation with or without differentiation become similar (which can be easily checked from formulations in Theorems~\ref{thm:opt} and \ref{thm:sp}).
In these two scenarios, similar resource allocations lead to similar revenues. These explains the bell shape for parameter $\bar{s}$.

Based on the above observations, we find that the revenue gain can be very high under two conditions. First, the high willingness to pay users are minorities in the effective market. Second, the total resource is comparatively limited.

For cases with three or more groups, the analytical study becomes much more challenging  due to many more parameters.  Moreover, the complex threshold structure of the effective market makes the problem even complicated. We will present some numerical studies to illustrate some interesting insights.

For illustration  convenience, we choose a three-group example and three different sets of parameters as shown in Table~\ref{tab:three_group}. To limit the dimension of the problem, we set the parameters such that the total number of users and the average willingness to pay (i.e., $\bar{\theta} = \sum_{i=1}^3 N_i\theta_i/(\sum_{i=1}^3  N_i)$) of all users are the same across three different parameter settings. This ensures that the $SP$ scheme achieves the same revenue in three different cases when resource is abundant. Figure~\ref{fig:threegroup} illustrates how the  differentiation gain changing changes in  resource $S$.

\begin{figure}[htb]
\begin{minipage}[c]{0.5\textwidth}
\centering
\tabcaption{Parameter settings of a three-group example}
\begin{tabular}{l|cc|cc|cc|c}
  \hline
   & $\theta_1$ & $N_1$ & $\theta_2$ & $N_2$ & $\theta_3$ & $N_3$ & $\overset{}{\bar{\theta}} $\\
  \hline
  Case 1 &9  & 10 & 3 & 10 &1  & 80 &2\\
  \hline
  Case 2 & 3 & 33 & 2 & 33 & 1 & 34 &2\\
  \hline
  Case 3 & 2.2 & 80 & 1.5 & 10 & 1 & 10 &2\\
  \hline
\end{tabular}
\label{tab:three_group}
\end{minipage}\\[4pt]
\begin{minipage}[c]{0.45\textwidth}
\centering
\includegraphics[scale=0.3]{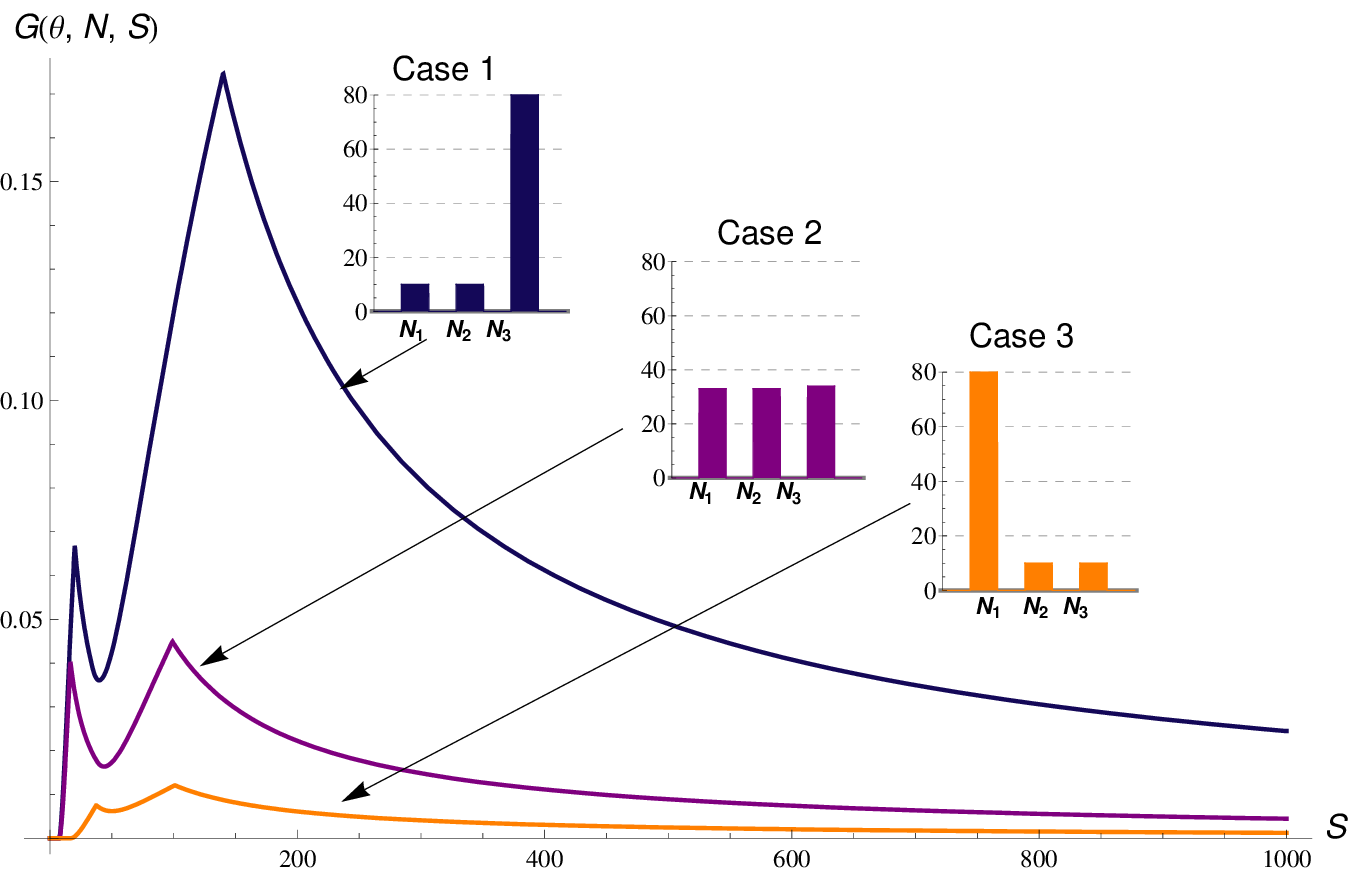}
\caption{An example of the revenue gain of a three-group market with the same average willingness to pay}
\label{fig:threegroup}
\end{minipage}
\end{figure}

Similar as the analytical study of the two-group case, Fig.~\ref{fig:threegroup} shows that the revenue gain is large only when the high willingness to pay users are minorities
(e.g. case 1) in the effective market and the resource is limited  but not too small ($100\le S\le 150$ in all three cases). When resource $S$ is large enough (e.g., $\ge 150$), the gain will gradually diminish to zero as the resource increases. For each curve in Fig.~\ref{fig:threegroup}, there are two peak points. Each peak point represents a change of the effective market threshold in the $SP$ scheme, i.e., when the resource allocation to a group becomes zero. In numerical studies of networks with $I>3$ groups  (not shown in this paper), we have observed the similar conditions for achieving a large differentiation gain and the phenomenon of $I-1$ peak points.

\subsection{What is the best tradeoff of Partial Price Differentiation?}
\label{sub:PPD_Numerical}
In Section~\ref{sec:PPD}, we design Algorithm~\ref{alg:3} that
optimally solves the $PP$ problem with a polynomial complexity.  Here we study the tradeoff between total revenue and implementational complexity.

To illustrate the tradeoff, we consider a five-group example with parameters shown in Table~\ref{tab:ppd}. Note that high willingness to pay users are minorities here. Figure~\ref{fig:ppd_example} shows the revenue gain $G$ as a function of total resource $S$ under different $PP$ schemes (including $CP$ scheme as a special case), and Fig.~\ref{fig:threshold} shows how the effective market thresholds change with the total resource.

\begin{table}[htb]
\centering \caption{Parameter setting of a five-group example}
\label{tab:ppd}
\begin{tabular}{|c|c|c|c|c|c|}
\hline
group  index $i$ &1&2&3&4&5\\
\hline \hline
$\theta_i$&16&8&4&2&1\\
\hline
$N_i$&2&3&5&10&80\\
\hline
\end{tabular}
\end{table}

We enlarge Fig.~\ref{fig:ppd_example} and Fig.~\ref{fig:threshold} within the range of $S\in[0,50]$, which is the most complex and interesting part due to several peak points.
Similar as Fig.~\ref{fig:threegroup}, we observe $I-1=4$ peak points for each curve in Fig.~\ref{fig:ppd_example}. Each peak point again represents a change of effective market threshold of the single pricing scheme, as we can easily verify by comparing Fig.~\ref{fig:threshold} with Fig.~\ref{fig:ppd_example}.

\begin{figure}[htb]
\centering
\includegraphics[scale=0.35]{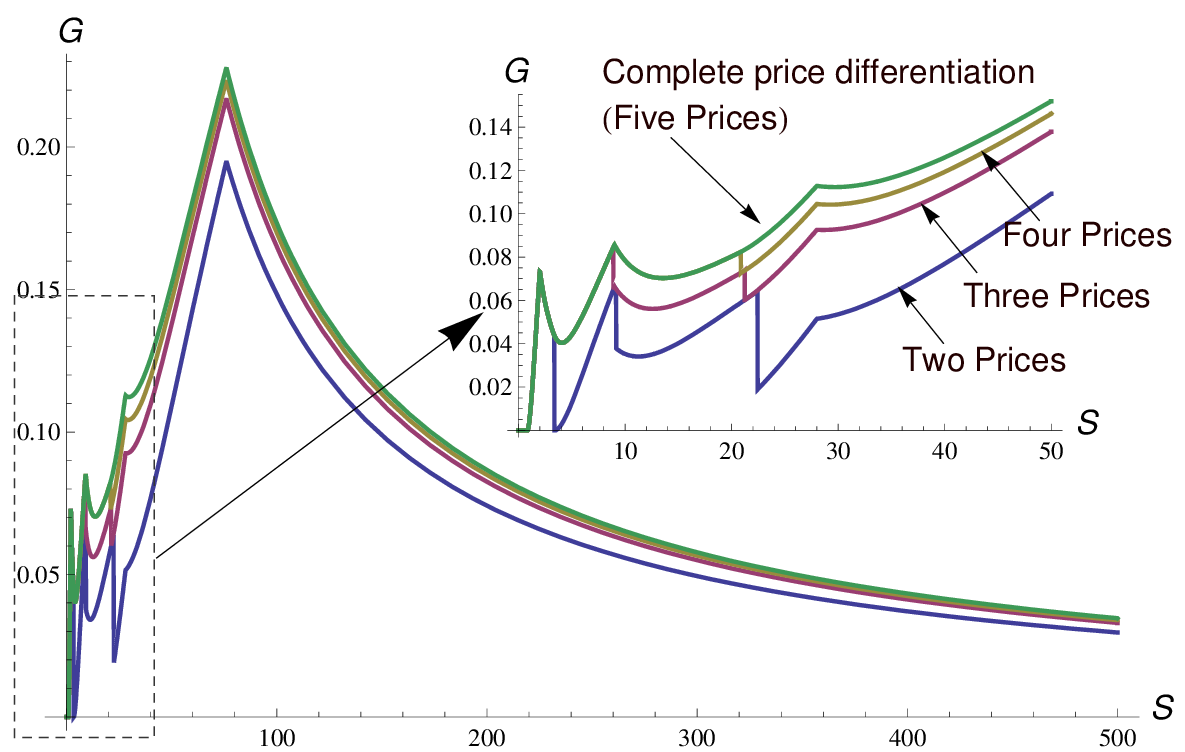}
\caption{Revenue gain of a five-group example under different price
differentiation schemes} \label{fig:ppd_example}
\vspace{3mm}
\includegraphics[scale=0.35]{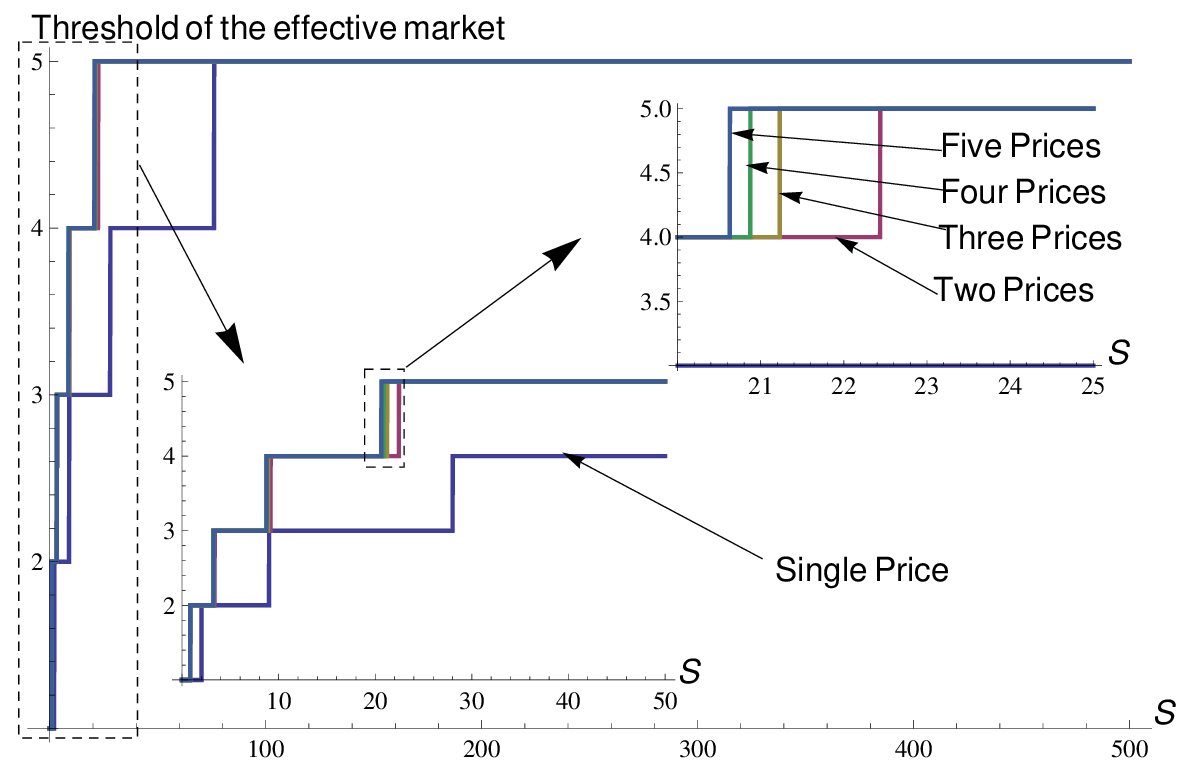}
\caption{Corresponding thresholds of effective markets of Fig.~\ref{fig:ppd_example}'s example} 
\label{fig:threshold}
\end{figure}

As the resource $S$ increases from $0$, all gains in Fig.~\ref{fig:ppd_example}  first overlap with each other, then the two-price scheme (blue curve) separates from the others at $S=3.41$, after that the three-price scheme (purple curve) separates at $S=8.89$, and finally the four-price scheme (dark yellow curve) separates  at near $S=20.84$.
These phenomena are due to the threshold structure of the $PP$ scheme. When the resource is very limited, the effective markets under all pricing scheme include only one group with the highest willingness to pay, and all pricing schemes coincide with the $SP$ scheme. As the resource increases, the effective market enlarges from two groups to finally five groups. The change of the effective market threshold can be directly observed in Fig.~\ref{fig:threshold}. Comparing across different curves in Fig.~\ref{fig:threshold}, we find that the effective market size is non-decreasing with the number of prices  for the same resource $S$. This agrees with our intuition in Section~\ref{sub_sp}, which states that the size of effective market indicates the degree of differentiation.

Figure~\ref{fig:ppd_example} provides the service provider a global picture of choosing the most proper pricing scheme according to achieve the desirable financial target under a certain parameter setting.
For example, if the total resource $S=100$, the two-price
scheme seems to be a sweet spot, as it achieves a differential gain of $14.8\%$ comparing to the $SP$ scheme and is only $2.4\%$ worse than the $CP$ scheme with five prices.

\section{Conclusion}
In this paper, we study the revenue-maximizing problem for a
monopoly \rev{service provider} under both complete and incomplete network information. Under complete information, our focus is to investigate the tradeoff between the  total revenue and the implementational complexity (measured in the number of pricing choices available for users). Among the three  pricing differentiation schemes we proposed (\ie complete,  single, and partial), the partial price differentiation is the most general one and includes the other two as special cases.
By exploiting the unique problem structure, we designed an algorithm that computes the optimal partial pricing scheme  in polynomial time, and numerically quantize the tradeoff between implementational complexity and total revenue.
Under incomplete information, designing an incentive-compatible  differentiation pricing scheme is difficult in general. We show that when the users are significantly different, it is possible to design a quantity-based pricing scheme that achieves the same maximum revenue as under complete information.



\begin{IEEEbiography}[{\includegraphics[width=1in,height=1.25in,clip,keepaspectratio]{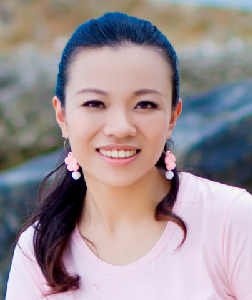}}]{Shuqin Li (S'09)} received Ph.D. in Information Engineering from The Chinese University of Hong Kong in 2012. She now works as a research scientist in Bell labs Shanghai, Alcatel-Lucent Shanghai Bell Co., Ltd. Her research interests include resource allocation, pricing and revenue management in communication networks, applied game theory, contract theory and incentive mechanism design in network economics, network coding and stochastic network optimization.
\end{IEEEbiography}

\begin{IEEEbiography}[{\includegraphics[width=1in,height=1.25in,clip,keepaspectratio]{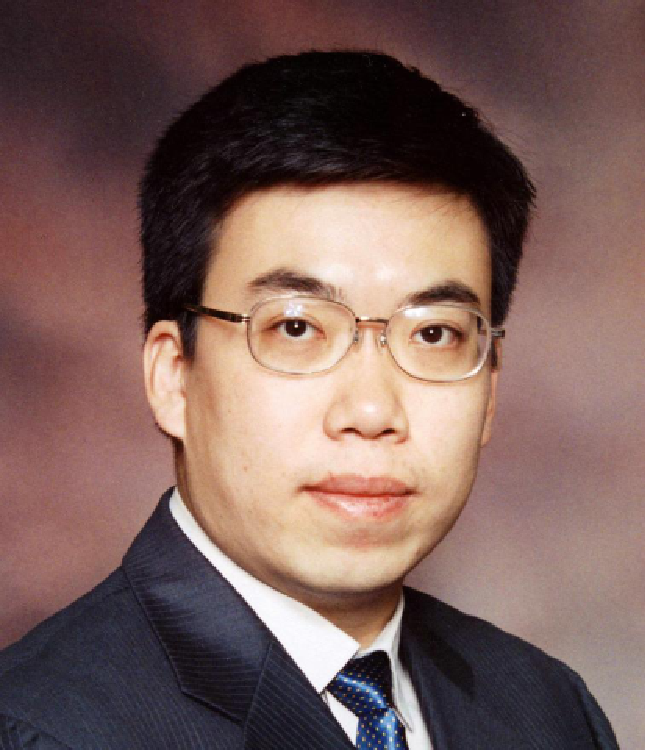}}]{Jianwei Huang (S'01-M'06-SM'11)}
is an Assistant Professor in the Department of Information Engineering at the Chinese University of Hong Kong. He received Ph.D. in Electrical and Computer Engineering from Northwestern University in 2005. He worked as a Postdoc Research Associate in the Department of Electrical Engineering at Princeton University during 2005-2007. 

Dr. Huang has served as the Editor of IEEE Journal on Selected Areas in Communications - Cognitive Radio Series, Editor of IEEE Transactions on Wireless Communications, and Guest Editor of IEEE Journal on Selected Areas in Communications and IEEE Communications Magazine. He is the Chair of IEEE ComSoc Multimedia Communications Technical Committee, a Steering Committee Member of IEEE Transactions on Multimedia and IEEE ICME. He has served as the TPC Co-Chair of IEEE GLOBECOM Selected Areas of Communications Symposium 2013, IEEE WiOpt 2012, IEEE ICCC Communication Theory and Security Symposium 2012, IEEE GlOBECOM Wireless Communications Symposium 2010, IWCMC Mobile Computing Symposium 2010, and GameNets 2009. He is the recipient of five Best Paper Awards, including the IEEE Marconi Prize Paper Award in Wireless Communications in 2011. He also received the IEEE ComSoc Asia-Pacific Outstanding Young Researcher Award in 2009. For more information, please see http://ncel.ie.cuhk.edu.hk. 
\end{IEEEbiography}

\appendix

\subsection{Complete Price Differentiation under complete information with General Utility Functions}
\label{sec:general_utility}
In this section, we extend the solution of complete price differentiation problem to general form of increasing and concave utility functions $u_i(s_i)$. We denote $R_i(s_i)$ as the revenue collected from one user in group $i$. Based on the stackelberg model, the prices satisfy $p_i=u_i'(s_i)$,
$s_i\ge 0$ $i\in\mathcal{I}$, thus
\begin{equation}
\label{eq:revenue}
    R_i(s_i)=u_i'(s_i)s_i, s_i\ge 0.
\end{equation}

Therefore, we can rewrite the Complete Price differentiation problem with General utility function ($CPG$) as follows.
\begin{eqnarray}
    CPG:& \underset{\boldsymbol{s}\geq 0,\boldsymbol{n}}{\text{maximize}}& \sum\limits_{i\in\mathcal{I}} n_i R_i(s_i) \nonumber\\
    &\text{subject to}&  n_i \in\{0,\ldots,N_{i}\} \;\;,\;\;  i\in\mathcal{I}\\
    &&\sum\limits_{i\in\mathcal{I}} n_i s_{i} \le S \label{eq:resource_con}
\end{eqnarray}

By similar solving technique in Sec.~III of the revised paper,
we can solve the $CPG$ problem by decomposing it into two subproblems: resource allocation subproblem $CPG_1$, and admission control subproblem $CPG_2$.
In subproblem $CPG_1$, for given $\boldsymbol{n}$, we solve
\begin{eqnarray*}
    CPG_1:& \underset{\boldsymbol{s}\geq 0}{\text{maximize}} & \sum\limits_{i\in\mathcal{I}} n_i R_i(s_i) \nonumber \\
    &\text{subject to}& \sum\limits_{i\in\mathcal{I}} n_i s_{i} \le S
\end{eqnarray*}
After solving the optimal resource allocation
$s^*_i(\boldsymbol{n})$, $i\in\mathcal{I}$, we further solve admission control
subproblem:
\begin{eqnarray*}
    CPG_2:& \underset{\boldsymbol{n}}{\text{maximize}} & \sum\limits_{i\in\mathcal{I}} n_i R_i(s^*_i(\boldsymbol{n})) \\
    &\text{subject to}:& n_i \in\{0,\ldots,N_{i}\}.
\end{eqnarray*}

We are especially interested in the case that constraint (\ref{eq:resource_con}) is active in the $CPG$ problem, which means the resource bound is tight in the considered problem; otherwise, the $CPG$ problem degenerates to a revenue maximization without any bounded resource constraint. We can prove the following results.
\begin{proposition}
\label{pro:general_utility_1} If the resource constraint
(\ref{eq:resource_con}) is active in the optimal solution of
the $CPG$ problem (or the $CPG_1$ subproblem), then  one of optimal
solutions of the $CPG_2$ subproblem is
\begin{equation}
    n_i^*=N_i,\;i\in \mathcal{I}.
\end{equation}
\end{proposition}

\begin{proof}
We first release the
variable $n_i$ to real number, and calculate the first
derivative as follows:
\begin{equation}
\frac{\partial R_i}{\partial n_i}=R_i(s^*_i)+n_i\,\frac{\partial
R_i(s^*_i)}{\partial s_i}\,\frac{\partial s^*_i}{\partial n_i},
\;\;i\in\mathcal{I}.
\end{equation}

Plugging (\ref{eq:revenue}),
$R'_i(s_i)=u_i''(s_i)\,s_i\,+\,u_i'(s_i)$, and we have
\begin{equation}
\label{eq:revenue_1} \frac{\partial R_i}{\partial
n_i}=u'_i(s_i^*)\left( s_i^*+n_i\,\frac{\partial s^*_i}{\partial
n_i}\right)+n_i\,u''_i(s_i^*)\,s_i^*\,\frac{\partial s^*_i}{\partial
n_i}, \; \;i\in\mathcal{I}.
\end{equation}

Since the resource constraint (\ref{eq:resource_con}) is active in
the optimal solution of the $CPG_1$ subproblem, that is,
$\sum\limits_{i\in\mathcal{I}} n_i s_{i} = S$, by taking derivative
of $n_i$ in both sides of it, we have:
\begin{equation}
\label{eq:derivative}
    s^*_i\,+\, n_i\frac{\partial s^*_i}{\partial n_i}= 0.
\end{equation}

Substituting (\ref{eq:derivative}) into (\ref{eq:revenue_1}), since we assume the
utility function $u_i(s_i)$ is increasing and concave function, then
we have
\begin{equation}
\frac{\partial R_i}{\partial n_i}=-u''_i(s_i^*)\,{s_i^*}^2\ge 0, \;
\;i\in\mathcal{K}.
\end{equation}

Thus we can conclude that one of optimal solutions for the $CPG_2$
subproblem is $n_i^*=N_i$, $i\in \mathcal{I}$.
\end{proof}

Proposition~\ref{pro:general_utility_1} points out that when the resource constraint (\ref{eq:resource_con}) is active, the $CPG$ problem can be greatly simplified: its solution can be obtained by solving the $CPG$ subproblem with
parameters $n_i=N_i$, $i=1,\dots, I$. The following proposition
provides a sufficient condition that the resource constraint (\ref{eq:resource_con}) is active.

\begin{proposition}
If $u''_i(s_i)s_i+u_i(s_i)>0$, $s_i\ge 0$, $i\in \mathcal{I}$, then
the resource constraint is active at the optimal solution.
\end{proposition}

\begin{proof}
Let $\lambda$ and $\mu_i$, $i\in\mathcal{I}$, be the Lagrange
multiplier of constraint (\ref{eq:resource_con}) and
$s_i\ge 0,i\in\mathcal{I}$ respectively, thus the KKT
conditions of the $CGP_1$ subproblem  is given as follows:
\begin{align*}
 n_i\frac{\partial R_i(s^*_i)}{\partial s_i}-n_i\lambda^*+\mu_i^*&=0, \, i\in\mathcal{I};\\
\lambda^*\left(\sum_{i\in\mathcal{I}}n_is_i^*-S\right)&=0;\\
\mu_i^*s_i^*&=0;\\
 \lambda^*&\ge 0;\\
 \mu_i^*&\ge 0,\, i\in\mathcal{I};\\
 s_i^*&\ge 0,\, i\in\mathcal{I}.
\end{align*}

We denote $\mathcal{K}:=\{i\,|\,s_i^*>0\}$, and
$\bar{\mathcal{K}}:=\{i\,|\,s_i^*=0\}$.

For $i\in \mathcal {K}$:
\begin{align}
 \frac{\partial R_i(s^*_i)}{\partial s_i}=\lambda^*,\; i\in\mathcal{I};\label{eq:KKT_derivative}\\
 \lambda^*\left(\sum_{i\in\mathcal{K}}n_is_i^*-S\right)=0. \label{eq:KKT_complementary}
\end{align}

For $i\in \bar{\mathcal {K}}$:
\begin{align}
 \frac{\partial R_i(0)}{\partial s_i}\le\lambda^*,\; i\in\mathcal{I};
\end{align}

Since $u''_i(s_i)s_i+u_i(s_i)>0$, $s_i\ge 0$, $i\in \mathcal{I}$ and
(\ref{eq:KKT_derivative}), we have
$$ \lambda^*=\frac{\partial R_i(s^*_i)}{\partial s_i}=u''_i(s^*_i)s^*_i+u_i(s^*_i)>0.$$

By (\ref{eq:KKT_complementary}), we must have
$\sum_{i\in\mathcal{I}}n_is_i^*-S=0$, that the resource constraint
is active at the optimal solution.
\end{proof}

Next, let us discuss how to calculate the optimal solution. To
guarantee uniqueness resource allocation solution, we assume that revenue in  is a strictly concave function of the demand\footnote{This assumption has been frequently used in the revenue management literature \cite{talluri2005theory}.}, \ie $\frac{\partial^2 R_i(s_i)}{\partial s_i^2}< 0$, $i\in \mathcal{I}$.
Thus we have the following theorem.
\begin{theorem}
\label{thm:general_utility}
If $\frac{\partial^2 R_i(s_i)}{\partial s_i^2}< 0$, $i\in \mathcal{I}$, then there exists an optimal solution of the $CGP$ problem as follows:
\begin{itemize}
    \item All users are admitted: $n^*_i=N_i$ for all $i\in\mathcal{I}.$
    \item There exist a value $\lambda^\ast$ and a group index threshold $K^{cp} \leq I$, such that only the top $K^{cp}$ groups of users receive positive resource allocations,
\begin{equation}
s_i^*=
    \begin{cases}
\frac{\partial R_i}{\partial s_i}^{-1}(\lambda^*),& i\in\mathcal{K}\,;\\
0, & \text{otherwise}.
     \end{cases}
\end{equation}where values of $\lambda^*$ and effective market $\mathcal{K}$ can  be computed as in Algorithm \ref{alg:general_utitlity}.
\end{itemize}
\end{theorem}

In Algorithm~\ref{alg:general_utitlity}, we use notation $f^{-1}$ denotes its inverse function, and rearrange the group index
satisfying $\frac{\partial R_{(1)}}{\partial s_{(1)}}^{-1}(0)\ge
\frac{\partial R_{(2)}}{\partial s_{(2)}}^{-1}(0)\ge\dots\ge
\frac{\partial R_{(I)}}{\partial s_{(I)}}^{-1}(0)$.
\begin{algorithm}[ht]
\caption{Search the threshold for general utility function}
\label{alg:general_utitlity}
\begin{algorithmic}[1]
\State $k\gets I$, $\lambda \gets \frac{\partial R_{(k)}}{\partial
s_{(k)}}^{-1}(0)$ \While{$\sum\limits_{i=1}^k \,
n_{(i)}\left(\frac{\partial R_{(i)}}{\partial
s_{(i)}}^{-1}(\lambda)\right)^+\,\ge\,S,$} \State $k \gets k-1$
\State $\lambda \gets \frac{\partial R_{(k)}}{\partial
s_{(k)}}^{-1}(0)$ \EndWhile \State \textbf{Return}
$\mathcal{K}=\{(1),(2),\dots,(k)\}$
\end{algorithmic}
\end{algorithm}

\begin{remark}
The complexity of Algorithm~\ref{alg:general_utitlity} is also  $\mathcal{O}(I)$,  \ie linear in the number of user groups (not the number of users).
\end{remark}

\begin{remark}
There are several functions satisfying the technical conditions in Theorem~\ref{thm:general_utility}, \eg the standard  $\alpha$-fairness functions
$$
u_i(s_i)=
    \begin{cases}
(1-\alpha)^{-1}s_i^{1-\alpha},& 0\le\alpha< 1;\\
\log s_i, & \alpha=1.
     \end{cases}
$$
\end{remark}


\subsection{Proof of Proposition~\ref{pro:comparison_sp_cd}}
\label{appendix_cp_sp}
\begin{proof}
We first focus on the key water-filling problems that we solve for the two
pricing schemes (the $CP$ scheme on the LHS and the $SP$ scheme on the RHS):
\begin{equation}
\label{eq:initial}
\sum_{i\in \mathcal{I}} N_i\left(\sqrt{\frac{\theta_i}{\lambda^*}}-1\right)^+
\;=\;S\;=\;
\sum_{i\in \mathcal{I}} N_i\left({\frac{\theta_i}{p^*}}-1\right)^+
\end{equation}

Let $\theta=\frac{{p^*}^2}{\lambda^*}$ be the solution of the equation of
$\sqrt{\frac{\theta}{\lambda^*}}=\frac{\theta}{p^*}$. By comparing it with $\theta_i$,
$i\in \mathcal{I}$, there are three cases:
\begin{itemize}
    \item Case 1:

$\theta> \theta_1\Rightarrow \sqrt{\frac{\theta_i}{\lambda^*}}\; = \; \frac{\sqrt{\theta_i}\sqrt{\theta}}{p^*}\; > \;\frac{\theta_i}{p^*},\; \forall\, i\in\mathcal{I}.$

 This case can not be possible. Since if every term in the left summation
 is strictly larger than its counterpart in the right
summation, then (\ref{eq:initial}) can not hold.

\item Case 2:
$\theta_I\ge \theta\Rightarrow\;\sqrt{\frac{\theta_i}{\lambda^*}}\; =
\;
\frac{\sqrt{\theta_i}\sqrt{\theta}}{p^*}\;\; \le\; \;\frac{\theta_i}{p^*},\; \forall\, i\in\mathcal{I}.$
Similarly as Case 1, it can not hold, either.

\item Case 3:
$\exists\,k,  \;\; s.t.\;1\le k<I\;{\rm{and}}\;\theta_k\ge\theta\ge\theta_{k+1}$
$$\Rightarrow\left\{{\begin{array}{ll}
\sqrt{\frac{\theta_i}{\lambda^*}}\;=&\underbrace{ \frac{\sqrt{\theta_i}\sqrt{\theta}}{p^*}\; \le \;\frac{\theta_i}{p^*}, i=1,2,\dots,k;}_\text{The equality holds only when $\theta=\theta_k$ and $i=k$.} \\
\\
\sqrt{\frac{\theta_i}{\lambda^*}} \;=&\underbrace{\frac{\sqrt{\theta_i}\sqrt{\theta}}{p^*}\; \ge \;\frac{\theta_i}{p^*},\;  i=k+1,\dots,I.}_\text{The equality holds only when $\theta=\theta_{k+1}$ and $i=k+1$.} \\
\end{array}} \right.$$
Similar argument as the above two case,  we have $K^{cp}\ge k$ and $K^{sp}\ge k$, otherwise
(\ref{eq:initial}) can not hold. Further, $K^{cp}\ge
K^{sp}$, since if $\frac{\theta_{K^{sp}}}{p^*}-1>0$, then
$\sqrt{\frac{\theta_{K^{cp}}}{\lambda^*}}-1>0$.
\end{itemize}
By Theorems~\ref{thm:opt} and \ref{thm:sp}, we prove the proposition.
\end{proof}

\subsection{Proof of Lemma~\ref{le:consecutive_index}}
\label{proof_le:consectutive_index}

We can first prove the following lemma.
\begin{lemma}
\label{le:add_one_group}
Suppose an effective market of the single pricing scheme is denoted as $\mathcal{K}=\{1,2,\dots,K\}$. If we add a new group $v$ of $N_v$ users with $\theta_v>\theta_K$, then the revenue strictly
increases.
\end{lemma}
\begin{proof}
We denote the single price before joining group~$v$ is $p$, the
price after joining group~$v$ is $p'$, the effective market become
$\mathcal{K}'$. By Theorem~\ref{thm:sp}, we have
$$p=\frac{\sum_{i=1}^K N_i\theta_i}{S+\sum_{i=1}^KN_i}\; \text{with}\;\;\theta_K>p\;\text{and}\;\;\theta_{K+1}\le p.$$
Since the optimal revenue is obtained by selling out the total
resource $S$, thus to prove the total revenue strictly increases if and only if we can prove $p'>p$. We consider the following two cases.
\begin{itemize}
    \item If after group~$v$ joining in, the new effective market satisfies $\mathcal{K}'=\mathcal{K}\cup \{v\}$, then we have
$$p'=\frac{\sum_{i=1}^K N_i\theta_i+N_v\theta_v}{S+\sum_{i=1}^KN_i+N_v}.$$
Since $\theta_v>\theta_K>p$, we have $p'>p$, due to the following
simple fact.
\begin{fact}
For any $a_1,\,b_1,\,a_2,\,b_2>0$, the following two inequality are
equivalent:
\begin{equation}
\label{eq:simple_fact}
\frac{a_1}{b_1}\ge\frac{a_2}{b_2}\Leftrightarrow
\frac{a_1}{b_1}\ge\frac{a_1+a_2}{b_1+b_2}\ge\frac{a_2}{b_2}.
\end{equation}
\end{fact}
     \item If after group~$v$ joining in, the new effective market shrinks, namely, $\mathcal{K}'\subset \mathcal{K}\cup \{v\},\;\mathcal{K}'\ne \mathcal{K}\cup \{v\}$, then we have $p'>\theta_{K}> p$.
\end{itemize}
\end{proof}

By the above Lemma~\ref{le:add_one_group}, we further prove Lemma~\ref{le:consecutive_index}.
\begin{proof}
We prove Lemma~\ref{le:consecutive_index} by contradiction. Suppose that the group indices of the
effective market under the optimal partition~$\boldsymbol{a}$ is not
consecutive. Suppose that group~$i$ is
not an effective group, and there exists some group~$j$, $j>i$,
which is an effective group. We consider a new
partition~$\boldsymbol{a}'$ by putting group~$i$ into the cluster to
which group~$j$ belongs, and keeping other groups unchanged.
According to Lemma~\ref{le:add_one_group}, the revenue under
partition~$\boldsymbol{a}'$ is greater than that under
partition~$\boldsymbol{a}$, thus partition~$\boldsymbol{a}$ is not optimal.
This contradicts to our assumption and thus completes the proof.
\end{proof}

\subsection{Proof of Theorem~\ref{th:consecutive}}
\label{proof_th:consecutive}
For convenience, we use the notation $(\cdots\cup\cdots|\cdots\cup\cdots|\cdots)$ to denote a partition with the groups between bars connected with ``$\cup$'' representing a cluster, e.g., three partitions for $J=2, K^{pp}=3$ are
$(1|2\cup3)$, $(1\cup 2\,|\,3)$ and $(1\cup3\,|\,2)$. In addition, we introduce the \emph{compound group} to simplify the notation of complex clusters with multiple groups. A cluster containing group $i$ can be simply represented as $Pre(i)\cup i\cup Post(i)$, where $Pre(i)$ (or $Post(i)$) refers as a compound group composing of all the groups with willingness to pay larger (or
smaller) than that of group~$i$ in the cluster. Note that the compound groups can be empty in certain cases.

Before we prove the general case in Theorem~\ref{th:consecutive},  we first prove the results is true for the following two special cases in Lemma~\ref{le:K_0=3} and Lemma~\ref{le:K=4}.
\begin{lemma}
\label{le:K_0=3}
For a three-group effective market with two prices, i.e., $K^{pp}=3$, $J=2$, an optimal partition involves consecutive group indices within clusters.
\end{lemma}
\begin{proof}
There are three partitions for  $K^{pp}=3$, $J=2$, and only $(1\cup
3\,|\,2)$ is with discontinuous group index within clusters. To show our result, we only need to prove one of partitions with group consecutive is better than $(1\cup 3\,|\,2)$.
We have two main steps in this proof, first we prove this result is true for $PP$ problem without consider Constraint (\ref{eq:pd_threshold_con}). Further, we show that Constraint (\ref{eq:pd_threshold_con}) will not affect the optimality of partitions with consecutive group indices within each cluster.

\textbf{\emph{Step 1:} (Without Constraint (\ref{eq:pd_threshold_con}))}

Without considering Constraint (\ref{eq:pd_threshold_con}), we want show that $\boldsymbol{a}_1=(1\cup 2\,|\,3)$ is
always better than $\boldsymbol{a}_2=(1\cup 3\,|\,2)$. Mathematically, what we try to
prove is:
\begin{equation}
v(\boldsymbol{a}_2)>v(\boldsymbol{a}_1).
\label{eq:12>13}
\end{equation}
where $v(\boldsymbol{a}_2)=(N_1+N_3)\sqrt{\frac{N_1\theta_1+N_3\theta_3}{N_1+N_3}}+N_2\sqrt{\theta_2}$, and
$v(\boldsymbol{a}_1)=(N_1+N_2)\sqrt{\frac{N_1\theta_1+N_2\theta_2}{N_1+N_2}}+N_3\sqrt{\theta_3}$.
With the new notation
$$\Delta V(i,j):=(N_i+N_j)\sqrt{\frac{N_i\theta_i+N_j\theta_j}
{N_i+N_j}}-N_i\sqrt{\theta_i}-N_j\sqrt{\theta_j},$$
it is easy to see that (\ref{eq:12>13}) is equivalent to the following inequality: 
\begin{equation}
\Delta V(1,3)>\Delta V(1,2). \label{eq:v_13}
\end{equation}

We prove the inequality (\ref{eq:v_13}) by considering the following
two cases.

a) If $N_1\le N_2$, we define a function of $x$ as follows,
\begin{eqnarray*}
g(j;x)\!:=\!(\!N_j\!+\!N_1\!)\!\sqrt{\!\frac{N_j\theta_j\!+\!N_1(\theta_j\!+\!x)}{N_j+N_1}}\!-\!N_1\!\sqrt{\theta_j\!+\!x}
\!-\!N_j\!\sqrt{\theta_j}.
\end{eqnarray*}
It is easy to check that
$$g(j;x)|_{x=\theta_{1}-\theta_{j}}=\Delta V(1,j),\;\text{and}\;\, g(j;x)|_{x=0}=0;$$
and if $x>0$, then
\begin{equation*}
\label{eq:closed_neighbor_D}
 g'(j;x)\!=\!\frac{\partial g(j;x)}{\partial x}\!=\!\frac{N_{1}}{2}  \left(\!\!\frac{1}{\sqrt{\frac{N_j\theta_j+N_1(\!\theta_j+x\!)}{N_j+N_1}}}\!-\!\frac{1}{\sqrt{\theta_{j}+x}}\!\!\right)\!> \!0,
\end{equation*}
\begin{equation*}
\text{and}\;\label{eq:closed_neighbor_D_D}
 \frac{\partial g'(j;x)}{\partial \theta_j}=\frac{N_{1}}{4}\!\! \left(\!\!\frac{1}{\left(\theta_{j}+x\right)\!^{1.5}}\!-\!\frac{1}{\left(\frac{N_j\theta_j+N_1(\theta_j+x)}{N_j+N_1}\right)\!^{1.5}}\!\!\right)\!\!<\!0.
\end{equation*}
Since $\theta_2>\theta_3$, it immediately follows that
$$    g'(3;x)\ge
\frac{N_{1}}{2}
\left(\frac{1}{\sqrt{\frac{N_3\theta_2+N_1(\theta_2+x)}{N_1+N_3}}}-\frac{1}{\sqrt{\theta_{2}+x}}\right),
$$
Since $N_2\ge N_1$, then we have
$$    g'(3;x)\ge
\frac{N_{1}}{2}\!
\left(\!\frac{1}{\sqrt{\frac{N_3\theta_2+N_1(\theta_2+x)}{N_1+N_3}}}-\frac{1}{\sqrt{\theta_{2}+x}}\!\right)\!\ge
g'(2;x),
$$
Thus, it follows
\begin{equation*}
    \Delta V(1,3)=\int^{\theta_1\!-\!\theta_3}_0\!\!\!\!\!\!\!\!\!g'(3;x)dx >\int^{\theta_1\!-\!\theta_2}_0\!\!\!\!\!\!\!\!\!g'(2;x)dx=\Delta V(1,2),
\end{equation*}
i.e., (\ref{eq:v_13}) is obtained.

Let us see a special case of  (\ref{eq:v_13}). When  $N_1=N_2$, then
$$\Delta V(1,2)\!=\!(N_1\!\!+\!\!N_1)\sqrt{\frac{N_1\theta_2\!\!+\!\!N_1\theta_1}{N_1+N_1}}\!-\!N_1\sqrt{\theta_1}\!-\!N_1\sqrt{\theta_2},$$
then we have
\begin{equation}
\label{eq:3group_N1N2}
    \Delta V(1,3)\!>\!(N_1\!\!+\!\!N_1)\sqrt{\frac{N_1\theta_2\!\!+\!\!N_1\theta_1}{N_1+N_1}}\!-\!N_1\sqrt{\theta_1}\!-\!N_1\sqrt{\theta_2}.
\end{equation}
Notice that although (\ref{eq:3group_N1N2}) is defined with the
assumption that $N_1\leq N_2$, it also holds for the case $N_1>N_2$
as (\ref{eq:3group_N1N2}) does not contain the parameter $N_2$. This
result will be used in the proof later.

b) If $N_1> N_2$, we define a function of $m$ as
$$f(m):=(N_1+m)\sqrt{\frac{N_1\theta_1+m\theta_2}{N_1+m}}-N_1\sqrt{\theta_1}-m\sqrt{\theta_2}.$$

It is easy to obtain that
\begin{eqnarray}
\frac{df(m)}{dm}
=\frac{\left(\sqrt{\frac{N_1\theta_1+m\theta_2}{N_1+m}}-\sqrt{\theta_2}\right)^2}{2\sqrt{\frac{N_1\theta_1+m\theta_2}{N_1+m}}}>0,\nonumber
\end{eqnarray}

i.e., the function $f$ is an increasing function of $m$.

Thus it follows that
\begin{equation*}
    \Delta V(1,2)=f(N_2)<f(N_1)\overset{(a)}{<}\Delta V(1,3),
\end{equation*}
where $(a)$ results from (\ref{eq:3group_N1N2}), the right hand side
of which is equal to $f(N_1)$.

\textbf{\emph{Step 2:} (Checking Constraint (\ref{eq:pd_threshold_con}))}

We want to prove that $\boldsymbol{a}_1$ satisfying Constraint (\ref{eq:pd_threshold_con}) is the sufficient condition of $\boldsymbol{a}_2$ satisfying (\ref{eq:pd_threshold_con}).

Consider if  $\boldsymbol{a}_1$ does not satisfy (\ref{eq:pd_threshold_con}), it means
\begin{equation*}
\sqrt{\theta_3}\le \sqrt{\lambda(\boldsymbol{a}_1)} =\sqrt{\frac{v(\boldsymbol{a}_1)}{S+\sum_{i=1}^3N_i}}.
\end{equation*}
By the result in Step 1, we know that $v(\boldsymbol{a}_1)<v(\boldsymbol{a}_2)$, then we have
\begin{equation*}
\sqrt{\theta_3}< \sqrt{\lambda(\boldsymbol{a}_2)} =\sqrt{\frac{v(\boldsymbol{a}_2)}{S+\sum_{i=1}^3N_i}},
\end{equation*}
and further
\begin{equation*}
{\theta_3}< \sqrt{\theta_3\lambda(\boldsymbol{a}_2)}<\sqrt{\frac{N_3\theta_3+N_1\theta_1}{N_1+N_3}\lambda(\boldsymbol{a}_2)} =\sqrt{\theta^1\lambda(\boldsymbol{a}_2)}.
\end{equation*}
It means $\boldsymbol{a}_2$ can not satisfy (\ref{eq:pd_threshold_con}) either. Thus we see that constraint (\ref{eq:pd_threshold_con}) actually does not affect the result in Step 1.
In conclusion, we show that in a simple case with $K^{pp}=3$, $J=2$,  an optimal partition involves consecutive group indices within clusters.
\end{proof}

Further, based on Lemma~\ref{le:K_0=3} we prove another simple special case.
\begin{lemma}
\label{le:K=4}
For a four-group effective market with two prices, i.e., $K^{pp}=4$, $J=2$, an optimal partition involves consecutive group indices within clusters.
\end{lemma}
\begin{proof}
For $K^{pp}=4$ and $J=2$ case, there are total seven possible partitions.
Three among them are with consecutive group index, $(1\,|\,2\cup
3\cup 4)$, $(1\cup 2\,|\,3\cup 4)$ and $(1\cup 2\cup 3\,|\,4)$. We
denote a set composed by these three partitions as $\Sigma_c$. We need to show the remaining four partitions are no better than some partition in $\Sigma_c$. To show this, we only need to transform them to some three-group case and apply the result of Lemma~\ref{le:K_0=3}.
\begin{itemize}
    \item \textit{Case 1}: $(1\cup 4,\,2\cup 3)$ is not optimal since we can prove $(1\cup 2\cup 3,4)\in \Sigma_c$ is better. To show it, we take $2\cup 3$ as a whole, then by Lemma~\ref{le:K_0=3}, it follows that $\Delta V(1,4)>\Delta V(1,\,2\cup 3)$.

     \item \emph{Case 2}: $(2,\, 1\cup 3\cup 4)$ is not optimal, since we can prove $(1\cup 2, 3\cup 4)\in\Sigma_c$ is better. To show it, we take $3\cup 4$ as a whole, then by Lemma~\ref{le:K_0=3}, it follows $\Delta V(1,\,3\cup 4)>\Delta V(1, 2).$

     \item \emph{Case 3}: $(3,\, 1\cup 2\cup 4)$ is not optimal, since we can prove $(1\cup 2\cup 3, 4)\in\Sigma_c$ is better. To show it, we take $1\cup 2$ as a whole, then by Lemma~\ref{le:K_0=3}, it follows that $\Delta V(1\cup 2,4)>\Delta V(1\cup 2,3).$

     \item \emph{Case 4}: $( 1\cup 3, 2\cup 4)$ is not optimal, since we can prove $(1\cup 2\cup 3, 4)\in\Sigma_c$ is better. To show it, by Lemma~\ref{le:K_0=3}, it follows that $\Delta V(2,4)>\Delta V(2,3)$, and that $\Delta V(1,3)\overset{(b)}{>}\Delta V(1,2\cup 3)$. Here inequality~(b) is also easily obtained, if we notice that $\theta_1>\theta_{2\cup 3}>\theta_3$, thus group~$2\cup 3$ can be also treated as the role of group~2 in Lemma~\ref{le:K_0=3}.
\end{itemize}
\end{proof}

Now Let us prove Theorem~\ref{th:consecutive}.
For convenience, we introduce the notation \emph{Compound group},
such as $Pre(i)$ or $Post(i)$, which represents some part of a
cluster with ordered group indices. For a group~$i$ in some
cluster, $Pre(i)$ (or $Post(i)$) refers as a compound group
composing of all the groups with willingness to pay larger (or
smaller) than that of group~$i$. For example, in a cluster
$1\cup 2\cup 3\cup 5\cup 7\cup 8$, $Pre(3)=1\cup 2$, $Post(3)=5\cup
7\cup 8$.  Note that compound groups can be empty, denoted as
$\emptyset$. In last example, $Pre(1)=Post(8)=\emptyset$. Since all
the groups within the compound group belong to one cluster, we
can apply  Lemma~\ref{le:consecutive_index}. For example, with the
previous cluster setting, $N_{Pre(3)}=N_1+N_2$, and
$\theta_{Pre(3)}=\frac{N_1\theta_1+N_2\theta_2}{N_1+N_2}$. By this
equivalence rule, a compound group actually has not much difference
with one original group. The conclusions of Lemma~\ref{le:K_0=3} and
Lemma~\ref{le:K=4} can be easily extended to compound groups.
\begin{proof}
Without loss of generality, suppose that the group indices order
within each cluster is increasing.

Now consider one partition with discontinuous group indices within
some clusters. We can check the group indices continuity for every
single group. For example, a group~$c$ belonging to a cluster~$\mathcal{C}$,
and its next neighbor in this cluster is group~$d$, if $c-d=1$,
then the group indices until $c$ are consecutive, and if $c-d>1$, then
the group indices are discontinuous, and we find a gap between $c$ and
$d$.

Suppose that checking group indices
continuity for each group following the increasing indices
order (or equivalently decreasing willingness to pay order) from group~$1$ to group~$K^{pp}$.
We do not find any gap until group~$u_1$ in cluster~$\mathcal{U}$. We denote group~$u_1$
next neighbor in cluster~$\mathcal{U}$ is group~$u_2$. Since there is a gap between $u_1$ and $u_2$, there exists a group
$v$ in another cluster $\mathcal{V}$ and satisfying $v=u_1+1<u_2$. Now we
can construct a better partition by rearranging the two clusters
$\mathcal{U}$ and $\mathcal{V}$, while keeping other clusters unchanged. We can view
$\mathcal{U}$ as $(Pre(u_2)\cup Post(u_1))$, and $\mathcal{V}$ as $(v\cup Post(v))$,
since there is no group before $v$ in super-group $\mathcal{V}$,  otherwise it
contradicts with the fact that we do not find any gap until group~$u_1$.
It is easy to show that there is some new partition better than
the original one by Lemma~\ref{le:K_0=3} and Lemma~\ref{le:K=4}.
There are two cases depending on whether $Post(v)$ is empty or
not. If $Post(v)=\emptyset$, according to Lemma~\ref{le:K_0=3}, we find another partition with $\mathcal{U}'=Pre(u_2)\cup v$, $\mathcal{V}'=Post(u_1)$
better than the original $\mathcal{U}$ and $\mathcal{V}$. If $Post(v)\ne\emptyset$, no matter
$\theta_{Post(v)}$ is larger than $\theta_{Post(u_2)}$ or not,
according to Lemma~\ref{le:K=4}, it is easy to construct other
partitions better than the original $\mathcal{U}$ and $\mathcal{V}$, since the compound groups in these original clusters
($\mathcal{U}=(Pre(u_2)\cup Post(u_1))$, and $\mathcal{V}=(v\cup Post(v))$) does not satisfy property of consecutive
group indices within each cluster.

In conclusion, we show that for general cases, if there is any gap in
the partition, then we can construct another partition that is
better, which is equivalent to that the optimal partition must satisfy
consecutive group indices within each cluster.
\end{proof}
\subsection{Proof of Theorem \ref{the:IC}}
\label{sub_appendix_IC}
\begin{proof}
Since $U_i(s,p_q)$ is a strictly increasing function in the interval $[0,s_i^*]$, then (\ref{IC2}) holds, if and only if the following inequality holds:
\begin{equation}
    U_i(s^*_q,p_q)\le U_i(s_{i\rightarrow q},p_q),\;\forall i<q.
\label{n_s}
\end{equation}
Since $t_{1q}>\dots>t_{Kq}$, (\ref{n_s}) can be simplified to

\begin{equation}
    t_{q-1q}^2\ln t_{q-1q}-(t_{q-1q}^2-1)+ \frac{\sum_{k=1}^{K}N_kt_{kq}}{\sum_{k=1}^{K}N_k+S}(t_{q-1q}-1)\ge 0,
\label{tij1_sim}
\end{equation}
where $t_{iq}=\sqrt{\frac{\theta_i}{\theta_q}}$. With a slight abuse of notation, we abbreviate $t_{q-1q}$ as $t_q$, ($q=2,\dots,K$) in the sequel.
It is easy to see that the following inequality is the necessary and sufficient condition of (\ref{tij1_sim}) for $q=2$, and sufficient condition of (\ref{tij1_sim}) for $q>2$:
\begin{equation}
    t_{q}^2\ln t_{q}-(t_{q}^2-1)+ \frac{t_{q}\sum_{k=1}^{q-1}N_k+N_q}{\sum_{k=1}^{K}N_k+S}(t_{q}-1)\ge 0.
\label{tij2}
\end{equation}
Let $g(t)$ be the left hand side of the inequality (\ref{tij2}). It is easy to check that $g(t)$ is a convex function, with $g(1)=0$, $g(\infty)=\infty$ and $g'(1)<0$. So there exists a root $t_q>1$. When $t>t_{q}$, the inequality (\ref{tij2}) holds, thus (\ref{n_s}) holds, and the conclusion in Theorem \ref{the:IC} follows.
\end{proof}

\end{document}